\documentclass[lettersize,journal]{IEEEtran}
\usepackage{amsmath,amsfonts,amssymb,graphicx, amsthm}

\usepackage{tmi}

\usepackage[export]{adjustbox}

\makeatletter
\newcommand{\bigcomp}{%
  \DOTSB
  \mathop{\vphantom{\sum}\mathpalette\bigcomp@\relax}%
  \slimits@
}
\newcommand{\bigcomp@}[2]{%
  \begingroup\m@th
  \sbox\z@{$#1\sum$}%
  \setlength{\unitlength}{0.9\dimexpr\ht\z@+\dp\z@}%
  \vcenter{\hbox{%
    \begin{picture}(1,1)
    \bigcomp@linethickness{#1}
    \put(0.5,0.5){\circle{1}}
    \end{picture}%
  }}%
  \endgroup
}
\newcommand{\bigcomp@linethickness}[1]{%
  \linethickness{%
      \ifx#1\displaystyle 2\fontdimen8\textfont\else
      \ifx#1\textstyle 1.65\fontdimen8\textfont\else
      \ifx#1\scriptstyle 1.65\fontdimen8\scriptfont\else
      1.65\fontdimen8\scriptscriptfont\fi\fi\fi 3
  }%
}
\makeatother

\def\eqref#1{(\ref{#1})}

\def\1{\bm{1}}

\DeclareMathAlphabet{\mathsfit}{\encodingdefault}{\sfdefault}{m}{sl}
\SetMathAlphabet{\mathsfit}{bold}{\encodingdefault}{\sfdefault}{bx}{n}

\newcommand{\E}{\mathbb{E}}

\newcommand{\R}{\mathbb{R}}

\DeclareMathOperator*{\argmin}{arg\,min}

\usepackage{booktabs}
\usepackage{wrapfig}
\usepackage{array}
\usepackage{tablefootnote}
\usepackage{subcaption}
\usepackage{pifont}
\usepackage{colortbl}
\usepackage{hyperref} 
\usepackage{bm}
\usepackage{xspace}
\usepackage{multirow}
\usepackage{float}

\DeclareMathAlphabet\mathbfcal{OMS}{cmsy}{b}{n}
\usepackage{tcolorbox}

\newcommand{\modl}{\textsc{MoDL}}
\newcommand{\us}{\textsc{SMUG}}

\newcommand{\x}{\mathbf{x}}
\newcommand{\w}{\mathbf{w}}
\newcommand{\y}{\mathbf{y}}

\newcommand{\A}{\mathbf{A}}

\newcommand{\mycomment}[1]{}
\newtheorem{lemma}{Lemma}

\usepackage{cite}

\def\BibTeX{{\rm B\kern-.05em{\sc i\kern-.025em b}\kern-.08em
    T\kern-.1667em\lower.7ex\hbox{E}\kern-.125emX}}
\begin{document}
\title{Robust MRI Reconstruction by Smoothed Unrolling (SMUG)}
%

\author{Shijun Liang*,~\IEEEmembership{Member, ~IEEE},~Van Hoang Minh Nguyen*, Jinghan Jia, ~\IEEEmembership{Student Member, ~IEEE},
Ismail Alkhouri, \IEEEmembership{Member,~IEEE},
Sijia~Liu,~\IEEEmembership{Senior~Member,~IEEE}, Saiprasad~Ravishankar,~\IEEEmembership{Senior~Member,~IEEE}
\thanks{*~Equal contribution.
S. Liang (corresponding author: \textit{liangs16@msu.edu}) is with the Biomedical Engineering (BME) Department at Michigan State University (MSU), East Lansing, MI, 48824, USA. M. Nguyen (\textit{nguye954@msu.edu}) is with the Mathematics Department at MSU.  J.Jia (\textit{jiajingh@msu.edu}) is with the Computer Science and Engineering (CSE) Department at MSU. I. Alkhouri (\textit{alkhour3@msu.edu} \& \textit{ismailal@umich.edu}) is with the Computational Mathematics, Science \& Engineering (CMSE) Department at MSU and the Electrical Engineering \& Computer Science Department at the University of Michigan, Ann Arbor, MI, 48109, USA. S. Liu (\textit{liusiji5@msu.edu}) is with the CSE Department at MSU.  S. Ravishankar (\textit{ravisha3@msu.edu}) is with the CMSE \& BME Departments at MSU.\\
\copyright~2025 IEEE. Personal use of this material is permitted. Permission
from IEEE must be obtained for all other uses, in any current or future media,
including reprinting/republishing this material for advertising or promotional
purposes, creating new collective works, for resale or redistribution to servers
or lists, or reuse of any copyrighted component of this work in other works.
} \vspace{-0.3in}}


%
%
%
%
\maketitle
%



\newcommand{\Def}[0]{\mathrel{\mathop:}=}

\newcommand{\bx}{\mathbf{x}}
\newcommand{\by}{\mathbf{y}}
\newcommand{\bz}{\mathbf{z}}
\newcommand{\bX}{\mathbf{X}}
\newcommand{\bh}{\mathbf{h}}
\newcommand{\bu}{\mathbf{u}}
\newcommand{\bg}{\mathbf{g}}
\newcommand{\din}{\mathcal D^{\texttt{tr}}}
\newcommand{\dint}{\tilde{\mathcal D}^{\texttt{tr}}}
\newcommand{\dout}{\mathcal D^{\texttt{val}}}
\newcommand{\doutt}{\tilde{\mathcal D}^{\texttt{val}}}
\newcommand{\btheta}{\boldsymbol{\theta}}
\newcommand{\bphi}{\boldsymbol{\phi}}
\newcommand{\bdelta}{\boldsymbol{\delta}}
\newcommand{\grad}{\nabla}
\newcommand{\egrad}{\widehat{\nabla}}
\newcommand{\task}{\mathcal{T}}
\newcommand{\rnn}{\text{{\tt RNN}}\xspace}
\newcommand{\EI}{\text{EI}}
\newcommand{\lopt}{\text{LO}}
\newcommand{\V}{\mathbb{V}}
\newcommand{\indep}{\perp \!\!\! \perp}
\newcommand{\pphi}[1]{\frac{\partial #1}{\partial \bphi}}

\newcommand{\pre}{\textsc{p}}
\newcommand{\ft}{\textsc{f}}
\newcommand{\adv}{\textsc{a}}
\newcommand{\sta}{\textsc{n}}
\newcommand{\super}{\textsc{s}}
\newcommand{\self}{\textsc{ss}}
\newcommand{\Sp}{\textit{supervised}}
\newcommand{\Ss}{\textit{self-supervised}}

\newcommand{\cL}{\mathcal{L}}
\newcommand{\cD}{\mathcal{D}}

\begin{abstract}

As the popularity of deep learning (DL) in the field of magnetic resonance imaging (MRI) continues to rise, recent research has indicated that DL-based MRI reconstruction models might be excessively sensitive to minor input disturbances, including worst-case or random additive perturbations. This sensitivity often leads to unstable aliased images. This raises the question of how to devise DL techniques for MRI reconstruction that can be robust to \textcolor{black}{these variations}.
To address this problem, we propose a novel image reconstruction framework, termed \textsc{\underline{Sm}oothed \underline{U}nrollin\underline{g}} ({\us}), which advances a deep unrolling-based MRI reconstruction model using a randomized smoothing (RS)-based robust learning approach. RS, which improves the tolerance of a model against input noise, has been widely used in the design of adversarial defense approaches for image classification tasks. Yet, we find that the conventional design that applies RS to the entire DL-based MRI model is ineffective. In this paper, we show that {\us} and its variants address the above issue by customizing the RS process based on the unrolling architecture of DL-based MRI reconstruction models. \textcolor{black}{We theoretically analyze the robustness of our method in the presence of perturbations.}
Compared to vanilla RS \textcolor{black}{and other recent approaches},
we show that {\us} improves the robustness of MRI reconstruction with respect to a diverse set of instability sources, including worst-case and random noise perturbations to input measurements, varying measurement sampling rates, and different numbers of unrolling steps. 
Our code is available at \texttt{\url{https://github.com/sjames40/SMUG_journal}}.

\end{abstract}
\section{Introduction}
\label{sec:intro}

Magnetic resonance imaging (MRI) is a popular noninvasive imaging modality, which involves a sequential and slow data collection.
As such, MRI scans can be accelerated by collecting
limited
data. In this case, the process of image reconstruction requires tackling an ill-posed inverse problem. To deliver accurate image reconstructions from such limited information, compressed sensing (CS)~\cite{lustig2008compressed} has been extensively used. 
Conventional CS-MRI assumes
the underlying image's sparsity (in practice, enforced in the wavelet domain~\cite{wave} or via total variation~\cite{totalv}). 
As further improvement to conventional CS, various learned sparse signal models have been well-studied, such as involving patch-based synthesis dictionaries~\cite{ravishankar2011dlmri,jacob2013blindCSMRI}, or sparsifying transforms~\cite{ravishankar2012learning,ravishankar2020review}.
Learned transforms have been shown to offer an efficient and effective framework for sparse modeling in MRI~\cite{wensaibres20}.\par

\textcolor{black}{Recently}, due to the outstanding representation power of convolutional neural networks (CNNs), they have been applied in single-modality medical imaging synthesis \textcolor{black}{and reconstruction} ~\cite{Schlemper2019Sigma-net:Reconstruction,Ravishankar2018DeepReconstruction,aggarwal2018modl,Schlemper2018AReconstruction}. The U-Net \textcolor{black}{architecture}, presented in \cite{Unet} and used in several studies, is a
popular deep CNN for many tasks involving image processing. They exhibit two key features: the use of a diminishing path for gathering contextual information, and a symmetric expansion path for precise localization.\par

Hybrid-domain DL-based image reconstruction methods, such as \underline{Mo}del-based reconstruction using \underline{D}eep \underline{L}earned priors
({\modl})~\cite{aggarwal2018modl}, \textcolor{black}{Iterative Shrinkage-Thresholding Algorithm (ISTA-Net)~\cite{zhang2018istanet}}, etc., are used to enhance stability and performance by ensuring data consistency 
in the training and reconstruction phases. In MR imaging, data consistency layers are often essential in reconstruction networks 
to ensure the image agrees with the measurement model~\cite{Zheng2019twodataconsist,casade2017deep}.
Various methods such as~\cite{sun2016deep,hammernik2018learning,zhang2018istanet,aggarwal2018modl} maintain this consistency by deep unrolling-based architectures, which mimic a traditional iterative algorithm and learn the associated regularization parameters. Other approaches ensure data consistency by applying methods such as denoising regularization~\cite{romanoRED17} and plug-and-play techniques~\cite{buzzard:18:pap}. Despite their recent advancements, DL-based MRI reconstruction models are shown to be vulnerable to tiny changes or noise in the input, shifts in the measurement sampling rate~\cite{antun2020instabilities, zhang2021instabilities}, and varying iteration numbers in unrolling schemes~\cite{gilton2021deep}. In such cases, the resulting images from DL models are of inferior quality which could possibly lead to inaccurate diagnoses and, consequently, undesirable clinical consequences.


It is of \textcolor{black}{much} importance in medical imaging applications to learn reconstruction models that are robust to various \textcolor{black}{measurement artifacts, noise, and} scan or data variations at test time.
Although there exist numerous robustification techniques ~\cite{madry2017towards,zhang2019theoretically,cohen2019certified,salman2020denoised} to tackle the instability of DL models in image classification tasks, methods to enhance the robustness of DL-based MRI reconstruction models are less explored due to their regression-based learning targets. 
Methods such as randomized smoothing (RS) and its variations~\cite{cohen2019certified, salman2020denoised,zhang2022robustify}, are often used in image classification. They diverge from traditional defense methods~\cite{madry2017towards,zhang2019theoretically} such as adversarial training, which provide some empirical robustness but are computationally expensive and could fail under more 
diverse perturbations. 
RS ensures the model's stability within a 
radius surrounding the input image~\cite{cohen2019certified}, which could be critical for medical use cases such as MRI. Recent early-stage research has begun to apply RS to DL-based MRI reconstruction in an end-to-end manner~\cite{wolfmaking}. However, the end-to-end RS approach might not always be an appropriate fit for all image reconstructors, such as physics-based and hybrid methods.



In our recent conference work~\cite{li2023smug}, we proposed integrating the RS approach within the {\modl} framework for the problem of MR image reconstruction. 
This is accomplished by using RS in each unrolling step and at the intermediate unrolled denoisers in {\modl}. This strategy is underpinned by the `pre-training + fine-tuning' technique \cite{zoph2020rethinking,salman2020denoised}. 
\textcolor{black}{This paper significantly expands over our conference work~\cite{li2023smug}, with added analysis, extension to multiple reconstruction models, and comprehensive experimental comparisons and ablation studies. 
We 
provide an analysis and conditions under which the proposed smoothed unrolling (SMUG) technique is robust against perturbations.
The analysis sheds light on robustness to additive perturbations and with respect to increasing unrolling steps in the reconstruction model.
Our work is the first to systematically integrate robustness operations into physics-based image reconstruction networks and provide both analysis and comprehensive empirical studies.}
Furthermore, we introduce a novel weighted smoothed unrolling scheme that learns image-wise weights during smoothing unlike conventional RS. This approach further improves the reconstruction performance.  \textcolor{black}{Furthermore, in this work, we evaluate worst-case additive perturbations in k-space or measurement space, in contrast to \cite{li2023smug}, where image-space perturbations were considered.}

\vspace{-0.1in}
\subsection{Contributions}
The main contributions of this work are as follows:
\begin{itemize}
    \item \textcolor{black}{We propose {\us} that  systematically integrates robustness operations (RS) into several physics-based unrolled image reconstruction networks.}
    \item  \textcolor{black}{We provide a theoretical analysis to demonstrate the robustness of {\us} for image reconstruction using the MoDL architecture.} 
    \item We enhance the performance of {\us} by introducing weighted smoothing as an improvement over conventional RS and 
    showcase the resulting gains.
    \item \textcolor{black}{We integrate the techniques into multiple unrolled models including {\modl}~\cite{aggarwal2018modl}, ISTA-Net~\cite{zhang2018istanet}, and E2E-VarNet~\cite{sriram2020end}
    and demonstrate improved robustness of our methods compared to the original schemes. We also show advantages for {\us} over end-to-end RS \cite{wolfmaking}, Adversarial Training (AT)~\cite{jia2022robustness}, Deep Equilibrium (Deep-Eq) models~\cite{Deep-eq}, Hierarchical Randomized Smoothing ~\cite{scholten2024hierarchical} and a leading diffusion-based model~\cite{chung2022score}. Extensive experiments demonstrate the potential of our approach in handling various types of reconstruction instabilities.}
    
\end{itemize}

\vspace{-0.1in}
\subsection{Paper Organization}
The remainder of the paper is organized as follows. In Section \ref{sec: preliminaries}, we present preliminaries and the problem statement. Our proposed method is described in Section \ref{method}. Section \ref{sec: experiment} presents experimental results and comparisons, and we conclude in Section \ref{sec:conclusion}. 

\section{Preliminaries and Problem Statement}
\label{sec: preliminaries}


\subsection{\textcolor{black}{Setup of MRI Reconstruction}}

Many medical imaging approaches involve
ill-posed inverse problems such as the work in~\cite{compress}, where the aim is to reconstruct the original signal $\mathbf x \in\mathbb{C}^{n}$ (vectorized image) from undersampled k-space 
measurements $\mathbf y \in  \mathbb{C}^m$ with $m < n$.  \textcolor{black}{Here, k-space~\cite{lin2000principles} refers to the measurement space in MRI, and is the spatial frequency domain of the acquired signal. In multi-coil MRI, different coils encode the signal differently according to their spatial sensitivity profiles.} The imaging system in MRI can be modeled as a linear system $\mathbf y \approx  \mathbf A \mathbf x $, {where $\mathbf A$ may take on different forms for single-coil or parallel (multi-coil) MRI, etc.}
For example, in the single coil Cartesian MRI acquisition setting, $\mathbf A =  \mathbf M \mathbf F$, where $\mathbf F$ is the 2D discrete Fourier transform and $\mathbf M$ is a masking operator that implements undersampling. 
With the linear observation model, 
 MRI reconstruction is often formulated as
\begin{align}
    \hat{\bx}=\underset{\bx}{\arg\min} ~ \|\mathbf A \bx - \by \|^{2}_2 + \lambda \mathcal{R}(\bx),
    \label{eq:inv_pro}
\end{align}
where $\mathcal{R}(\cdot)$ is a regularization function (\textit{e.g.}, $\ell_1$ norm in the wavelet domain to impose a sparsity prior \textcolor{black}{\cite{wave}}), and $\lambda > 0$ is the regularization parameter.  






{\modl} \cite{Aggarwal2019MoDL:Problems} is a recent popular supervised deep learning approach inspired by the MR image reconstruction optimization problem in \eqref{eq:inv_pro}. {\modl} combines a denoising network with a data-consistency (DC) module in each iteration of an unrolled architecture.
In {\modl}, the hand-crafted regularizer, $\mathcal R$, is replaced by a learned network-based prior $\left \| \mathbf{x} - \cD_{\boldsymbol \theta}(\mathbf{x}) \right \|_{2}^2$ involving a  network $\cD_{\boldsymbol \theta}$. 
MoDL attempts to optimize this loss by initializing $\bx^0 = \mathbf{A}^H \mathbf{y}$, and then iterating the following process for a number of unrolling steps indexed by $n\in \{0,\dots,N-1\}$. Specifically, {\modl} iterations are given by 
%
%
\begin{equation}
    \label{modleqn1}
    \x^{n+1} = 
     \argmin_{\x} \|\A \x - \y\|_2^2
    + \lambda \|\x- \mathcal D_{\boldsymbol{\theta}}(\x^n)\|_2^2.
\end{equation}
\textcolor{black}{Here, the \textbf{Denoising Step} is given by
\(\boldsymbol{z^n} = \mathcal 
 D_{\boldsymbol{\theta}}(\x^n)\)} 
 \textcolor{black}{and the \textbf{Data Consistency (DC) Step} is given by
\[
\x^{n+1} = \argmin_{\boldsymbol{z}} \|\A\x - \y\|^2 + \lambda \|\x - \boldsymbol{z^{n}}\|^2.
\]
The DC step has a closed-form solution given by
\[
\x^{n+1} = (\A^H \A + \lambda \boldsymbol{I})^{-1} \big(\A^H \y + \lambda \boldsymbol{z}^n\big).
\]
The solution is implemented using conjugate gradients (CG).}
After $N$ iterations, we denote the final output of {\modl} as $\x^N = \boldsymbol{F}_{\text{MoDL}}(\x^0)$. The weights of the denoiser are shared across the $N$ blocks and are learned in an end-to-end supervised manner \cite{aggarwal2018modl}.

\subsection{\textcolor{black}{Lack of Robustness of DL-based Reconstructors}}

In \cite{antun2020instabilities}, it was demonstrated that deep learning-based MRI reconstruction can exhibit instability, when faced with subtle, nearly imperceptible input perturbations. These perturbations are commonly referred to as `adversarial perturbations' and have been extensively investigated in the context of DL-based image classification tasks, as outlined in~\cite{Goodfellow2015explaining}. In the context of MRI, these perturbations represent the worst-case additive perturbations, which \textcolor{black}{can be} used to evaluate method sensitivity and robustness~\cite{antun2020instabilities,jia2022on,10447906}. Let $\boldsymbol \delta$ denote a small perturbation of the measurements that falls in an $\ell_\infty$ ball of radius $\epsilon$, \textit{i.e.}, $\| \boldsymbol \delta \|_\infty \leq \epsilon$. Adversarial disturbances then correspond to the worst-case input perturbation vector $\boldsymbol \delta$ that maximizes the reconstruction error, \textit{i.e.}, 
\begin{equation}
\begin{array}{ll}
\max_{\| \bdelta \|_\infty \leq \epsilon }       \|  
\boldsymbol{F}_\text{{\modl}} (\mathbf A^H (\mathbf y + \boldsymbol \delta)) -  \mathbf t \|_2^2,
\end{array}
  \label{eq: atk_perturb}
\end{equation}
where $\mathbf t$ is a ground truth target image from the training set (\textit{i.e.}, label). The operator $\mathbf A^H $  transforms the measurements $\mathbf y$ to the image domain, and $\mathbf A^H \mathbf y$ is the input (aliased) image to the reconstruction model.
The optimization problem in \eqref{eq: atk_perturb} can be effectively solved using the iterative projected gradient descent (PGD) method~\cite{madry2017towards}. 

In \textbf{Fig.~\ref{fig: weakness}-(a)} and \textbf{(b)}, we show reconstructed images using {\modl} originating from a benign (i.e., undisturbed) input and a PGD-perturbed input, respectively. It is evident that the worst-case input disturbance significantly deteriorates the quality of the reconstructed image. While 
one focus of this work is to enhance robustness against input perturbations, \textbf{Fig.\ref{fig: weakness}-(c)} and \textbf{(d)} highlight two additional potential sources of instability that the reconstructor ({\modl}) can encounter during testing: variations in the measurement sampling rate (resulting in ``perturbations'' to the sparsity of the sampling mask in $\mathbf A$) \cite{antun2020instabilities}, and changes in the number of unrolling steps \cite{gilton2021deep}. In scenarios where the sampling mask (\textbf{Fig.\ref{fig: weakness}-(c)}) or number of unrolling steps (\textbf{Fig.\ref{fig: weakness}-(d)}) deviate from the settings used during {\modl} training, we observe a significant degradation in performance compared to the original setup (\textbf{Fig.\ref{fig: weakness}-(a)}), even in the absence of additive measurement perturbations. In Section~\ref{sec: experiment}, 
we demonstrate how our method
improves the reconstruction robustness in the presence of different types of perturbations, including those in \textbf{Fig.\ref{fig: weakness}}.

\begin{figure}[t]
\centering
    \includegraphics[width=9cm]{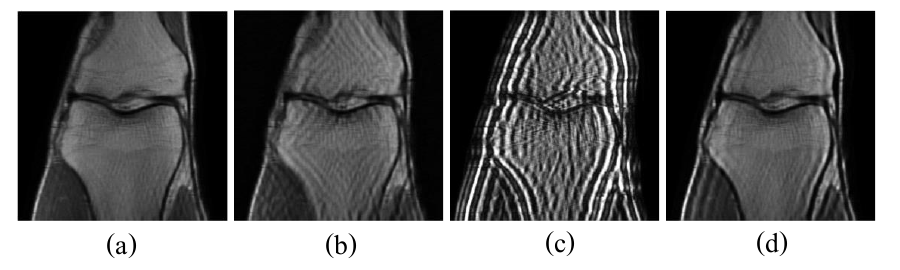}
\vspace{-0.3cm}
\caption{{\modl}'s instabilities resulting from perturbations to input data, the measurement sampling rate, and the number of unrolling steps used at testing phase shown on an image from the \texttt{fastMRI} 
dataset~\cite{zbontar2018fastmri}. We refer readers to Section~\ref{sec: experiment} for further details about the experimental settings. (a) {\modl} reconstruction from benign (\textit{i.e.}, without additional noise/perturbation) measurements with $ 4\times$ acceleration (\textit{i.e.}, 25\% sampling rate) and 8 unrolling steps. (b) {\modl} reconstruction from disturbed input with perturbation strength $\epsilon = 0.02$ (see Section~\ref{sec:experiment:setup}). (c) {\modl} reconstruction from clean measurements with $ 2\times$ acceleration (\textit{i.e.}, 50\% sampling), and using 8 unrolling steps. (d) {\modl} reconstruction from clean or unperturbed measurements with $ 4\times$ acceleration and \emph{16} unrolling steps.
\textcolor{black}{In (b), (c), and (d), the network trained in (a) is used.}
}
\label{fig: weakness}
\vspace{-0.5cm}
\end{figure}

\subsection{\textcolor{black}{Randomized Smoothing (RS)}}

Randomized smoothing, introduced in \cite{cohen2019certified}, enhances the robustness of DL models against noisy inputs.
It is implemented by generating multiple randomly modified versions of the input data and subsequently calculating an averaged output from this diverse set of inputs. 

Given some function $f(\mathbf x)$, RS formally replaces $f$ with a smoothed version
\begin{equation}
    \label{eqn: RS}
    g(\mathbf x) \Def \mathbb{E}_{\boldsymbol \eta \sim\mathcal{N}(\mathbf 0, \sigma^2\mathbf{I})} [ f( \bx + \boldsymbol \eta) ]\:,
\end{equation}
where $\mathcal{N}(\mathbf 0, \sigma^2\mathbf{I})$ denotes a Gaussian distribution with zero mean and element-wise variance $\sigma^2$, and $\mathbf{I}$ denotes the identity matrix of appropriate size. Prior research has shown that RS has been effective as an adversarial defense approach in DL-based image classification tasks \cite{cohen2019certified, salman2020denoised, zhang2022how}. However, the question of whether RS can significantly improve the robustness of {\modl} and other image reconstructors has not been thoroughly explored. A preliminary investigation in this area was conducted by \cite{wolfmaking}, which demonstrated the integration of RS into MR image reconstruction in an end-to-end (E2E) setting.
We can formulate image reconstruction using RS-E2E as
\begin{align}
    \x_{\textsc{RS-E2E}}=\mathbb{E}_{\boldsymbol \eta \sim\mathcal{N}(\mathbf 0, \sigma^2\mathbf{I})} [ 
    \boldsymbol{F}_{\text{\modl}} (\mathbf A^H (\mathbf{y} + \boldsymbol \eta)) ]
    .
    \label{eq: denoised smoothing mri}
    \tag{RS-E2E}
\end{align}
\textcolor{black}{This formulation aligns with the one used in~\cite{wolfmaking}, where the random noise vector $\boldsymbol \eta$ is directly added to $\mathbf{y}$ in the frequency domain (complex-valued), followed by multiplication with $\mathbf A^H$ to obtain the input image for \modl. The noisy measurements are also utilized in each iteration in \modl. RS-E2E can be identically formulated for alternative reconstruction models.}





\textbf{Fig.\,\ref{fig: combined}} shows a block diagram of \ref{eq: denoised smoothing mri}-backed {\modl}. 
This RS-integrated {\modl} is trained with supervision in the standard manner.
Although \eqref{eq: denoised smoothing mri} represents a straightforward application of RS to {\modl}, it remains unclear if this formulation is the most effective method to incorporate RS into unrolled algorithms such as {\modl}, considering the latter's 
specialties, e.g., the involved denoising and the data-consistency (DC) steps.\par

As such, for the rest of the paper, we focus on studying the following questions \textbf{(Q1)}--\textbf{(Q4)}. 
 \begin{tcolorbox}[left=1.2pt,right=1.2pt,top=1.2pt,bottom=1.2pt]
\textbf{(Q1)}: \textit{How should RS be integrated into an unrolled algorithm such as {\modl}?} \\
\textbf{(Q2)}: \textit{How do we learn the network $\mathcal{D}_{\boldsymbol{\theta}}(\cdot)$ in the presence of RS operations?} \\
\textbf{(Q3)}: \textit{Can we prove the robustness of SMUG in the presence of data perturbations?} \\
\textbf{(Q4)}: \textit{Can we further improve the RS operation in SMUG for enhanced image quality or sharpness?}
\end{tcolorbox}


\newtheorem{theorem}{Theorem}

\section{Methodology}
\label{method}
In this section, we address questions \textbf{(Q1)}--\textbf{(Q4)} by taking the unrolling characteristics of {\modl} into the design of an RS-based MRI reconstruction. The proposed novel integration of RS with {\modl} is termed 
{\textsc{\underline{Sm}oothed \underline{U}nrollin\underline{g}}}
(\textbf{\us}). \textcolor{black}{We also explore an extension of {\us} with a new weighted smoothing that yields improved performance.} \textcolor{black}{We note that while we primarily develop our methods based on {\modl}, in Section~\ref{sec: SMUG into other unrolled NWs} and Section~\ref{istanetresults}, we discuss extension to other unrolling methods such as ISTA-Net and E2E-VarNet.}



\subsection{Solution to (\textbf{Q1}): RS at intermediate unrolled denoisers}
\label{sec: unrolling}
As illustrated in \textbf{Fig.\ref{fig: combined}} (top), the RS operation in RS-E2E is typically applied to {\modl} in an end-to-end manner. 
This does not shed light on which component of {\modl} needs to be made more robust.
Here, we explore integrating RS at each intermediate unrolling step of {\modl}.
In this subsection, we present {\us}, which applies RS to the denoising network. 
This seemingly simple modification is related to a robustness certification technique known as ``denoised smoothing'' \cite{salman2020denoised}. In this technique, a smoothed denoiser is used, proving to be sufficient for establishing robustness in the model. We use $ \x_{\textrm{S}}^{n}$ to \textcolor{black}{denote} the $n$-th iteration of {\us}. Starting from $\x_{\textrm{S}}^{0} = \A^H \y$, the procedure is given by 
\begin{equation}
    \x_{\textrm{S}}^{n+1} = \argmin_{\x} \|\A \x - \y\|_2^2+
     \lambda \|\x- \mathbb{E}_{\boldsymbol{\eta}}\big[ \mathcal{D}_{\boldsymbol{\theta}}(\x_{\textrm{S}}^n +\boldsymbol{\eta}) \big] \|_2^2\:,
\end{equation}
%
where $\boldsymbol{\eta}$ is drawn from $\mathcal{N}(\mathbf 0, \sigma^2\mathbf{I})$. \textcolor{black}{After $N$ iterations, the final output of SMUG is denoted by $\x^N_{\text{S}} = \boldsymbol{F}_{\text{SMUG}}(\x^0)$, where $\boldsymbol{F}_{\text{SMUG}}(\cdot)$ denotes the end-to-end mapping.}
\textcolor{black}{The middle row of \textbf{Fig.\,\ref{fig: combined}} presents the architecture of {\us}.}

\begin{figure*}[t]
    \centering
    \makebox[\textwidth][c]{\includegraphics[width=0.78\textwidth]{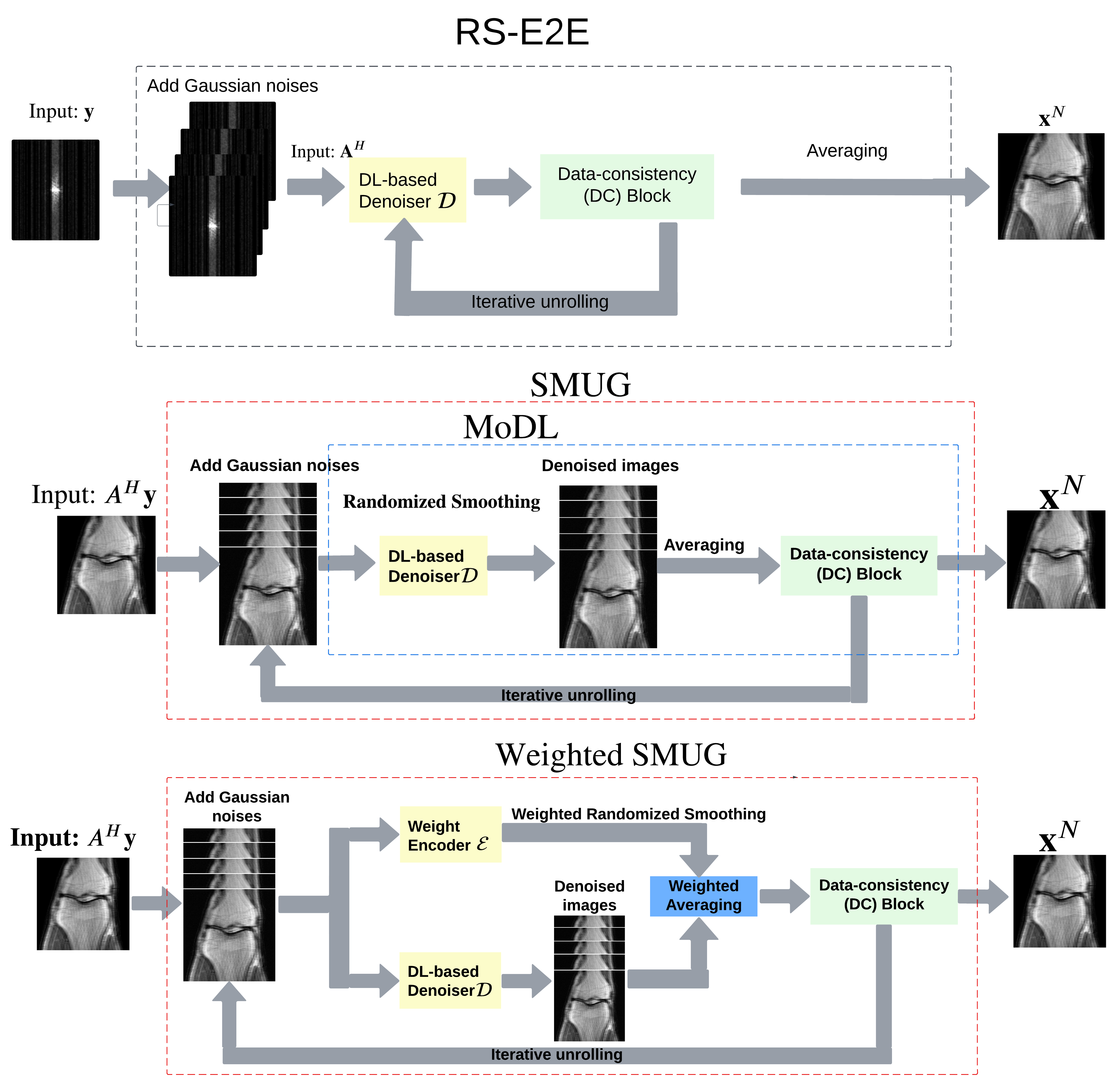}}
    \vspace*{-2em}
    \caption{\textcolor{black}{The three randomized smoothing-based architectures for reconstruction. In \textbf{RS‐E2E}, we generate \(N\) noisy k-space versions by adding Gaussian noise to \(\mathbf{y}\) and then apply the Hermitian operator \(\mathbf{A}^H\) to obtain samples that are batch-processed by a neural network for initial denoising. These outputs are refined by a data consistency module using the closed-form update~\eqref{modleqn1}, and after a few unrolled iterations, the final reconstruction is obtained by averaging the outputs. In contrast, the \textbf{{\us}} architecture directly adds random Gaussian noise in the image domain to create multiple noisy versions that are denoised by the neural network; their averaged output serves as a randomized smoothing step before applying the same data consistency module, yielding the final smoothed result after several iterations. Extending this framework, \textbf{Weighted {\us}} employs a learned weighted averaging obtained from a weighted encoder applied prior to the data consistency step—to produce the final smoothed reconstruction after a few unrolled iterations.}}
    \label{fig: combined}
    \vspace*{-4mm}
\end{figure*}

\subsection{Solution to \textbf{(Q2)}: {\us}'s pre-training \& fine-tuning}
\label{sec: how to smoothing}
In this subsection, we develop the training scheme of {\us}. Inspired by the currently celebrated ``pre-training + fine-tuning'' technique \cite{zoph2020rethinking,salman2020denoised}, we propose to train {\us} following this learning paradigm. Our rationale is that pre-training can provide a robustness-aware initialization of the DL-based denoising network for fine-tuning. To pre-train the denoising network $\mathcal D_{\btheta}$, we consider a mean squared error (MSE) loss that measures the Euclidean distance between images denoised by $\mathcal D_{\btheta}$ and 
the target (\textcolor{black}{ground truth}) images, denoted by $\mathbf t$. This leads to the \textbf{pre-training} step:
\begin{equation}
   \btheta_\mathrm{pre} = \displaystyle \argmin_{\btheta} \mathbb{E}_{\mathbf t \in \mathcal T} [ \mathbb{E}_{\boldsymbol \eta }  || \mathcal D_{\boldsymbol \theta} (  \mathbf t + \boldsymbol \eta)  -  \mathbf t ||_2^2]\:,
\label{eq: pre-train_loss}
\end{equation}
where $\mathcal T $ is the set of ground truth images in the training dataset. Next, we develop the fine-tuning scheme to improve $\btheta_\mathrm{pre}$ based on the labeled/paired MRI dataset. Since RS in {\us}, i.e., \textbf{Fig.~\ref{fig: combined}} (middle), is applied to every unrolling step, we propose an \textit{unrolled stability (\textbf{UStab}) loss} for fine-tuning 
$\mathcal{D}_{\btheta}$: 
\begin{align}
    \ell_{\mathrm{UStab}}(\btheta;   \mathbf y, \mathbf t)=
    \sum_{n=0}^{N-1} \mathbb{E}_{\boldsymbol \eta }||\cD_{\btheta}( \bx^n+\boldsymbol \eta)-\cD_{\btheta}(\mathbf t)||^2_2\:.
    \label{eq: unrolling loss}
\end{align}
The UStab loss in \eqref{eq: unrolling loss} relies on the  target images.
\textcolor{black}{The regularization exploits the target to better guide the behavior of the denoiser with random noise perturbations in each unrolling iteration to ensure enhanced stability of denoising.
It would appear more intuitive to use 
$\mathbf t$ instead of $D_{\btheta}(\mathbf t)$ inside the loss to directly minimize target estimation error. However, our study in \textbf{Fig.}~\ref{fig: unrolling_loss} using different loss configurations indicate that the former option degrades robustness and using $D_{\btheta}(\bx^n)$ or $D_{\btheta}(\mathbf t)$ in \eqref{eq: unrolling loss} to match to denoised unperturbed inputs or denoised targets yields more stable models.}

Integrating the UStab loss, defined in \eqref{eq: unrolling loss}, with the standard reconstruction loss, 
we obtain the \textbf{fine-tuned} $\btheta$ 
by minimizing
$ \mathbb{E}_{(\mathbf y, \mathbf t)} [ \ell(\btheta; \mathbf y, \mathbf t)]$, where
\begin{equation}
   \displaystyle 
   \ell(\btheta; \mathbf y, \mathbf t)= \ell_{\mathrm{UStab}}(\btheta;   \mathbf y, \mathbf t) + \lambda_\ell
   \| 
   \boldsymbol{F}_{\text{SMUG}}
   (\mathbf A^H \mathbf y) - \mathbf t \|_2^2    
   , 
\label{eq: finetune_loss}
\end{equation}
with $\lambda_\ell > 0$ \textcolor{black}{representing} a regularization parameter to strike a balance between the reconstruction error (for accuracy) and the denoising stability (for robustness) terms. We initialize $\btheta$ as $\btheta_\mathrm{pre}$ when optimizing \eqref{eq: finetune_loss} using standard optimizers such as Adam~\cite{kingma2014adam}.

\textcolor{black}{In practice, the same dataset is used for fine-tuning as pre-training because the pre-trained model is initially trained solely as a denoiser, while the fine-tuning process aims at integrating the entire regularization strategy applied to the MoDL framework. This approach ensures that the fine-tuning optimally adapts the model to the specific enhancements introduced by our robustification strategies.}




\subsection{Answer to \textbf{(Q3)}: Analyzing the robustness of SMUG in the presence of data perturbations}
\label{sec: why smoothing}


The following theorem discusses the robustness (i.e., sensitivity to input perturbations) achieved with {\us}. Note that all norms on vectors (resp. matrices) denote the $\ell_2$ norm (resp. spectral norm) unless indicated otherwise.

\textcolor{black}{
\begin{theorem}
Assume the denoiser network's output is bounded in norm.
Given the initial input image $\mathbf{A}^H \mathbf{y}$ obtained from measurements $\mathbf{y}$, let the SMUG 
reconstructed image at the $n$-th unrolling step be  $\mathbf{x}^n_{\text{S}}(\mathbf{A}^H \mathbf{y})$ 
with RS variance of $\sigma^2$. Let $\boldsymbol{\delta}$ denote an additive perturbation to the measurements $\mathbf{y}$. 
Then, 
\begin{equation}
    \label{eqn: Th 1 tight bound 1}
    \|\mathbf{x}^n_{\text{S}}(\mathbf{A}^H \mathbf{y}) - \mathbf{x}^n_{\text{S}}(\mathbf{A}^H (\mathbf{y}+\boldsymbol{\delta}))\| \leq C_n \|\boldsymbol{\delta}\|,
\end{equation}
where $C_n = \alpha \|\mathbf{A}\|_2    \begin{pmatrix} \frac{1-\begin{pmatrix}\frac{M \alpha}{\sqrt{2 \pi}\sigma}\end{pmatrix}^{n}   }{1 - \frac{M \alpha}{\sqrt{2 \pi}\sigma} } \end{pmatrix} +  \|\mathbf{A}\|_2   \begin{pmatrix} \frac{M \alpha}{\sqrt{2 \pi}\sigma} \end{pmatrix}^{n}$, with $\alpha =  \|(\mathbf{A}^H\mathbf{A}+\mathbf{I})^{-1}\|_2$ and $M = 2\max_\x(\|\mathcal{D}_{\boldsymbol\theta}(\x)\|)$.
%
\end{theorem}}
%



\textcolor{black}{
The proof is provided in the Appendix.
Note that the output of SMUG $\mathbf{x}^n_{\text{S}}(\cdot)$ depends on both the initial input (here $\mathbf{A}^H \mathbf{y}$) and the measurements $\mathbf{y}$. We abbreviated it to $\mathbf{x}^n_{\text{S}}(\mathbf{A}^H \mathbf{y})$ in the theorem and proof for notational simplicity.
The constant $C_n$ depends on the number of iterations or unrolling steps $n$ as well as the RS standard deviation parameter $\sigma$. 
For large $\sigma$, the robustness error bound for SMUG clearly decreases as the number of iterations $n$ increases.
In particular, if $\sigma > M \alpha/\sqrt{2 \pi}$, then as $n \rightarrow \infty$, $C_n \rightarrow \alpha \left \| \mathbf{A} \right \|_2/\begin{pmatrix}
1 - \frac{M\alpha}{\sqrt{2 \pi \sigma}}
\end{pmatrix}$.
Furthermore, as $\sigma \to \infty$, $C_n \to C \triangleq \alpha \left \| \mathbf{A} \right \|_2$.
Clearly, if $\alpha \leq 1$ and  $\left \| \mathbf{A} \right \|_2 \leq 1$ (normalized), then $C \leq 1$.}

\textcolor{black}{Thus, for sufficient smoothing, the error introduced in the SMUG output due to input perturbation never gets worse than the size of the input perturbation.
Therefore, the output is stable with respect to (w.r.t.) perturbations.
These results corroborate experimental results in Section~\ref{sec: experiment} on how SMUG is robust (whereas other methods, such as vanilla {\modl}, breakdown) when increasing the number of unrolling steps at test time, and is also more robust for larger $\sigma$ (with good accuracy-robustness trade-off). 
Also, the only assumption in our analysis is that the denoiser network output is bounded in norm. This consideration is handled readily when the denoiser network incorporates bounded activation functions such as the sigmoid or hyperbolic tangent. Alternatively, if we expect image intensities to lie within a certain range, a simple clipping operation in the network output would ensure boundedness for the analysis. 
The boundedness assumption is different from a non-expansiveness requirement; instead, it forms the basis for proving stability. The randomized smoothing (RS) component plays a pivotal role in ensuring the robustness bound, as it integrates smoothing into every unrolling step, stabilizing the outputs against input perturbations.}

A key distinction between SMUG and prior works, such as RS-E2E~\cite{wolfmaking}, is that smoothing is performed in every iteration. Moreover, while~\cite{wolfmaking} assumes the end-to-end mapping is bounded, in {\modl} or SMUG, it clearly isn't because the data-consistency step's output is unbounded as $\mathbf{y}$ grows.

\textcolor{black}{We remark that our intention with Theorem~1 is to establish a baseline of robustness intrinsic to models with unrolling architectures.}

\subsection{Solution to \textbf{(Q4)}: Weighted Smoothing}
In this subsection, we present a modified formulation of randomized smoothing to improve its performance in SMUG. Randomized smoothing in practice involves uniformly averaging images denoised with random perturbations. This can be viewed as a type of mean filter, which may lead to oversmoothing of structural information in practice. As such, we propose weighted randomized smoothing, which employs an encoder to assess a weighting (scalar) for each denoised image and subsequently applies the optimal weightings while aggregating images to enhance the reconstruction performance. 
\textcolor{black}{The approach with its image-adaptive smoothing mechanism could better combine image features based on their quality (see Fig.~\ref{weight check} later). Improved smoothing approaches could hold key value for image reconstruction problems, where the generated image is often directly evaluated.}
Our method not only surpasses the SMUG technique but also excels in enhancing image sharpness across various types of perturbation sources. This 
allows for a more versatile or flexible and effective approach for improving image quality under different conditions.

The weighted randomized smoothing operation applied on a function $f(\cdot)$ is as follows:
\begin{equation}
\label{eqn: WRS}
g_{\textrm{w}}(\x) := \frac{\mathbb{E}_{\boldsymbol \eta } [ w( \mathbf x + \boldsymbol \eta)
     f( \mathbf x + \boldsymbol \eta) ]}{\mathbb{E}_{\boldsymbol \eta } [ w( \mathbf x + \boldsymbol \eta) ] } \:,
\end{equation}
where $w(\cdot)$ is an input-dependent weighting function.

Based on the weighted smoothing in \eqref{eqn: WRS},
we introduce \textbf{Weighted \us} (\textbf{Fig.\,\ref{fig: combined}} bottom row). This approach involves \textcolor{black}{applying} weighted RS at each denoising step, and the weighting encoder is trained in conjunction with the denoiser during the fine-tuning stage. \textcolor{black}{For the weighting encoder in our experiments, we use a simple architecture consisting of five successive convolution, batch normalization, and ReLU activation layers followed by a linear layer and Sigmoid activation.} Specifically, in the $n$-th unrolling step, we use a weighting encoder $\mathcal E_{\boldsymbol \phi }$, parameterized by $\boldsymbol{\phi}$, to learn the weight of each image used for (weighted) averaging. Here, we use $\x^{n}_{\textrm{W}}$ to denote the output of the $n$-th block. Initializing $\x^{0}_{\textrm{W}} = \A^H \y$, the output of Weighted {\us} w.r.t. $n$ is
\begin{equation}
\begin{split}
    &\x_{\textrm{W}}^{n+1} = \argmin_{\x}~ \|\A \x - \y\|_2^2~+
    \\
    & \lambda     \begin{Vmatrix} \x- \frac{\mathbb{E}_{\boldsymbol \eta } [ \mathcal E_{\boldsymbol \phi } ( \mathbf x_{\textrm{W}}^{n} + \boldsymbol \eta) \mathcal{D}_{\boldsymbol \theta} ( \mathbf x_{\textrm{W}}^{n} + \boldsymbol \eta) ]}{\mathbb{E}_{\boldsymbol \eta } [ \mathcal E_{\boldsymbol \phi } ( \mathbf x_{\textrm{W}}^{n} + \boldsymbol \eta) ] } \end{Vmatrix}_2^2\:.
\end{split}
\end{equation}
%


After $N$ iterations, the final output of Weighted SMUG is $\x^N_{\text{W}} = \boldsymbol{F}_{\text{wSMUG}}(\x^0)$.
%
%
Figure~\ref{fig: combined} bottom row illustrates the \textcolor{black}{block diagram} of weighted {\us}. 
%
%

Furthermore, we extend the ``pre-training+fine-tuning'' approach proposed in Section~\ref{sec: how to smoothing} to the Weighted SMUG method.
In this case, we obtain the \textbf{fine-tuned} $\boldsymbol{\theta}$ and $\boldsymbol{\phi}$ by using
\begin{equation}
   \displaystyle \min_{\boldsymbol \theta, \boldsymbol \phi} \, \mathbb{E}_{(\mathbf y, \mathbf t)} [ \lambda_l \| 
   {F}_{\text{wSMUG}}
   (\mathbf A^H \mathbf y) - \mathbf t \|_2^2  +    \ell_{\mathrm{UStab}}(\boldsymbol \theta ;   \mathbf y, \mathbf t) ].
\label{eq: finetune_loss_weighted}
\end{equation}

\subsection{\textcolor{black}{Integrating RS into Other Unrolled Networks}}\label{sec: SMUG into other unrolled NWs}

\textcolor{black}{In this subsection, we further discuss the extension of our {\us} schemes for other unrolling based reconstructors, using ISTA-Net \cite{zhang2018istanet} and E2E-VarNet\cite{sriram2020end} as an example. The goal is to demonstrate the generality of our proposed approaches for deep unrolled models.}


\textcolor{black}{ISTA-Net uses a training loss function composed of discrepancy and constraint terms. In particular, it performs the following for $N$ unrolling steps: 
\begin{equation}\label{eqn: ista step 1}
    \boldsymbol{r}^n = \x^{n-1} -\lambda^{(n)} \A^H(\A \x^{n-1} -\y)
\end{equation}
\begin{equation}\label{eqn: ista step 2}
    \x^n = \mathcal{\hat{F}}^n (\textbf{Soft}(\mathcal{F}^n(\boldsymbol{r}^n) ,\theta^n))\:,
\end{equation}
where $\mathcal{\hat{F}}$ and $\mathcal{F}$ involve two linear convolutional layers (without bias terms) separated by ReLU activations, and $\mathcal{\hat{F}}^n \circ \mathcal{F}^n$ are constrained close to the identity operator. The function $\textbf{Soft}$ performs soft-thresholding with parameter $\theta^n$~\cite{zhang2018istanet}.
}

\textcolor{black}{Similar to SMUG for MoDL, we integrate RS into the network-based regularization (denoising) component of ISTA-Net.
This results in the following modification to \eqref{eqn: ista step 2}: 
\begin{equation}\label{eqn: ista SMUG step 2}
    \x^n = \mathbb{E}_{\boldsymbol{\eta}} [\mathcal{\hat{F}}^n (\textbf{Soft}(\mathcal{F}^n(\boldsymbol{r}^n+\boldsymbol{\eta}) ,\theta^n))]\:,
\end{equation}
where $\boldsymbol{\eta}$ is drawn from $\mathcal{N}(\mathbf 0, \sigma^2\mathbf{I})$. For weighted SMUG, \eqref{eqn: ista step 2} becomes
\begin{equation}\label{eqn:istawSMUGstep2}
    \x^n = \frac{\mathbb{E}_{\boldsymbol{\eta}} [\mathcal E_{\boldsymbol \phi } ( \boldsymbol{r}^{n} + \boldsymbol \eta)\mathcal{\hat{F}}^n (\textbf{Soft}(\mathcal{F}^n(\boldsymbol{r}^n+\boldsymbol{\eta}), \theta^n))]}{\mathbb{E}_{\boldsymbol{\eta}} [\mathcal E_{\boldsymbol \phi } ( \boldsymbol{r}^{n} + \boldsymbol \eta)]}\:.
\end{equation}
}

\textcolor{black}{We explore extending SMUG to an additional unrolling reconstructor, E2E-VarNet.}
\textcolor{black}{E2E-VarNet 
unrolls the following iteration for $N$ steps with updates performed in the measurement space:}
\textcolor{black}{\begin{equation}
\boldsymbol k_{t+1} = \boldsymbol k_t - \boldsymbol \eta_t \, \boldsymbol M \bigl(\boldsymbol k_t - \boldsymbol {\tilde{k}}\bigr) \;+\; \boldsymbol G\bigl(\boldsymbol k_t\bigr),
\end{equation}}
\textcolor{black}{where $\boldsymbol G$ is the refinement or denoising regularization module given by
\begin{equation}
\boldsymbol G\bigl(\boldsymbol k_t\bigr)
=
\boldsymbol F \circ \boldsymbol S \circ \textbf{CNN}\bigl(\boldsymbol S^{-1} \circ \boldsymbol F^{-1}(\boldsymbol k_t)\bigr).
\end{equation}}
\textcolor{black}{Here, $\textbf{CNN}$ is any parametric function that takes a complex
image as input and maps it to another complex image. Since it is applied after combining all coils into a single complex image, the same network can be used for scans with different numbers of coils. $\boldsymbol S$ and $\boldsymbol F$ denote coil-wise sensitivity weighting and Fourier transform, respectively, and $`\circ'$ denotes composition.}

\textcolor{black}{We integrate SMUG with E2E-VarNet by the following modification: 
\begin{equation}
\boldsymbol G\bigl(\boldsymbol k_t\bigr)
=
\boldsymbol F \circ \boldsymbol S \circ \mathbb{E}_{\boldsymbol{\eta}} [\textbf{CNN}\bigl(\boldsymbol S^{-1} \circ \boldsymbol F^{-1}(\boldsymbol k_t) + \boldsymbol{\eta} \bigr)],
\end{equation}
where $\boldsymbol{\eta}$ is drawn from $\mathcal{N}(\mathbf 0, \sigma^2\mathbf{I})$. 
The extension with Weighted SMUG is done similar to the case in~\eqref{eqn:istawSMUGstep2}.}



\section{Experiments}
\label{sec: experiment}
\begin{figure*}
    \centering
    \includegraphics[width=0.88\textwidth]{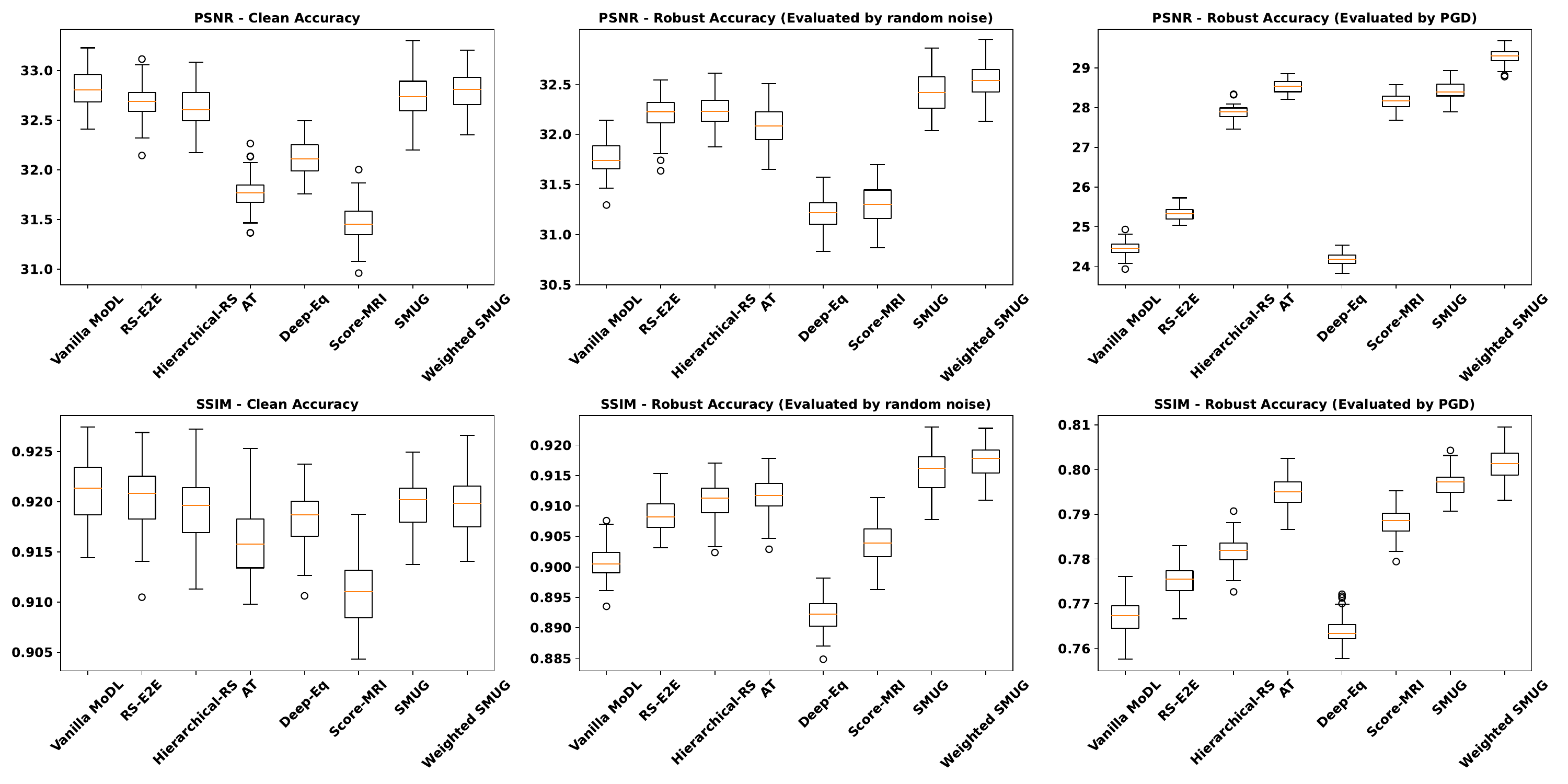}
    \vspace{-0.35cm}
    \caption{{ \textcolor{black}{Reconstruction accuracy box plots for the fastMRI \textbf{brain}  dataset with 4x acceleration factor. The additive random Gaussian noise of the second column plots is obtained using standard deviation of $0.01$. The worst-case additive noise of the third  column is obtained using the PGD method with $\epsilon = 0.02$}. }}
    \label{fig: box plot knee}
\end{figure*}
\begin{figure*}[!t]
\centering
\begin{tabular}[b]{cccc}
    \textbf{\small Ground Truth}& 
    \textbf{\small Vanilla {\modl}}&
    \textbf{\small RS-E2E}&
    \textbf{\small \textsc{SMUG}}\\  
    \includegraphics[width=.13\linewidth,valign=t]{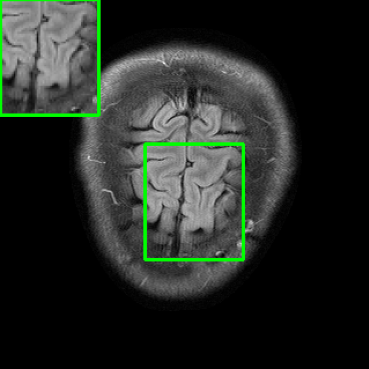}&
    \includegraphics[width=.13\linewidth,valign=t]{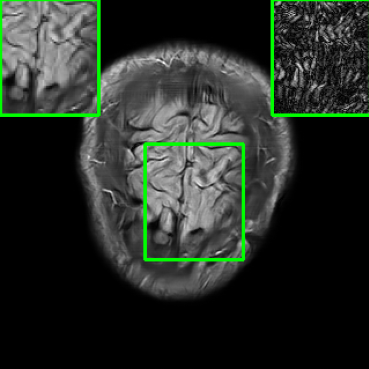}&
   \includegraphics[width=.13\linewidth,valign=t]{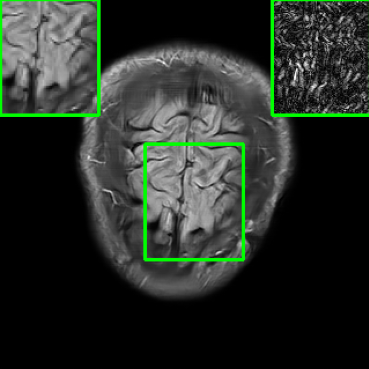} &
    \includegraphics[width=.13\linewidth,valign=t]{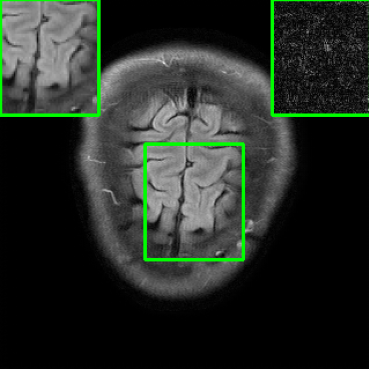}\\
    \scriptsize{PSNR = $\infty$ dB} & \scriptsize{PSNR = 24.84 dB} & \scriptsize{PSNR = 25.78 dB}& \scriptsize{\textbf{PSNR = 30.81 dB}}\\
    \textbf{\small AT}&
    \textbf{\small Score-MRI}&
    \textbf{\small Deep-Equilibrium}&
    \textbf{\small Weighted-SMUG}\\  
    \includegraphics[width=.13\linewidth,valign=t]{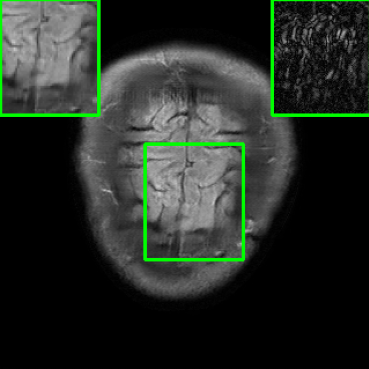} &
   \includegraphics[width=.13\linewidth,valign=t]{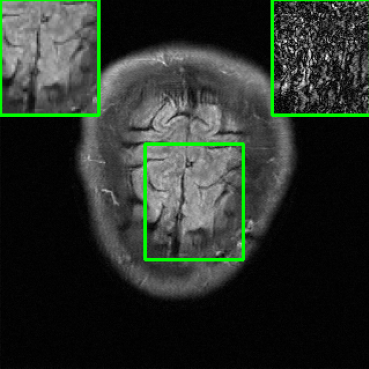} &
    \includegraphics[width=.13\linewidth,valign=t]{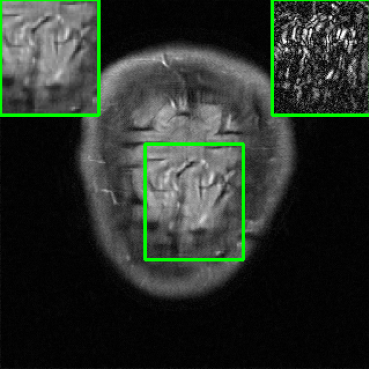} &
     \includegraphics[width=.13\linewidth,valign=t]{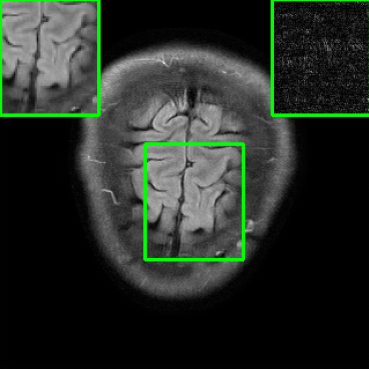}\\
         \scriptsize{PSNR = 30.72 dB} & \scriptsize{PSNR = 30.21 dB} & \scriptsize{PSNR = 24.58 dB}& \scriptsize{\textbf{PSNR = 31.41 dB}}\\
\end{tabular}
\caption{\textcolor{black}{Visualization of ground truth and reconstructed images using different methods for 4x k-space undersampling, evaluated on PGD-generated worst-case inputs of perturbation strength $\epsilon = 0.02$. \textcolor{black}{The reconstruction PSNRs are shown with the best values bolded.}}}
\label{fig:denoised_imgs_zoomed_300case}
\vspace{-0.0 in}
\end{figure*}

\begin{figure*}
    \centering
    \includegraphics[width=0.88\textwidth]{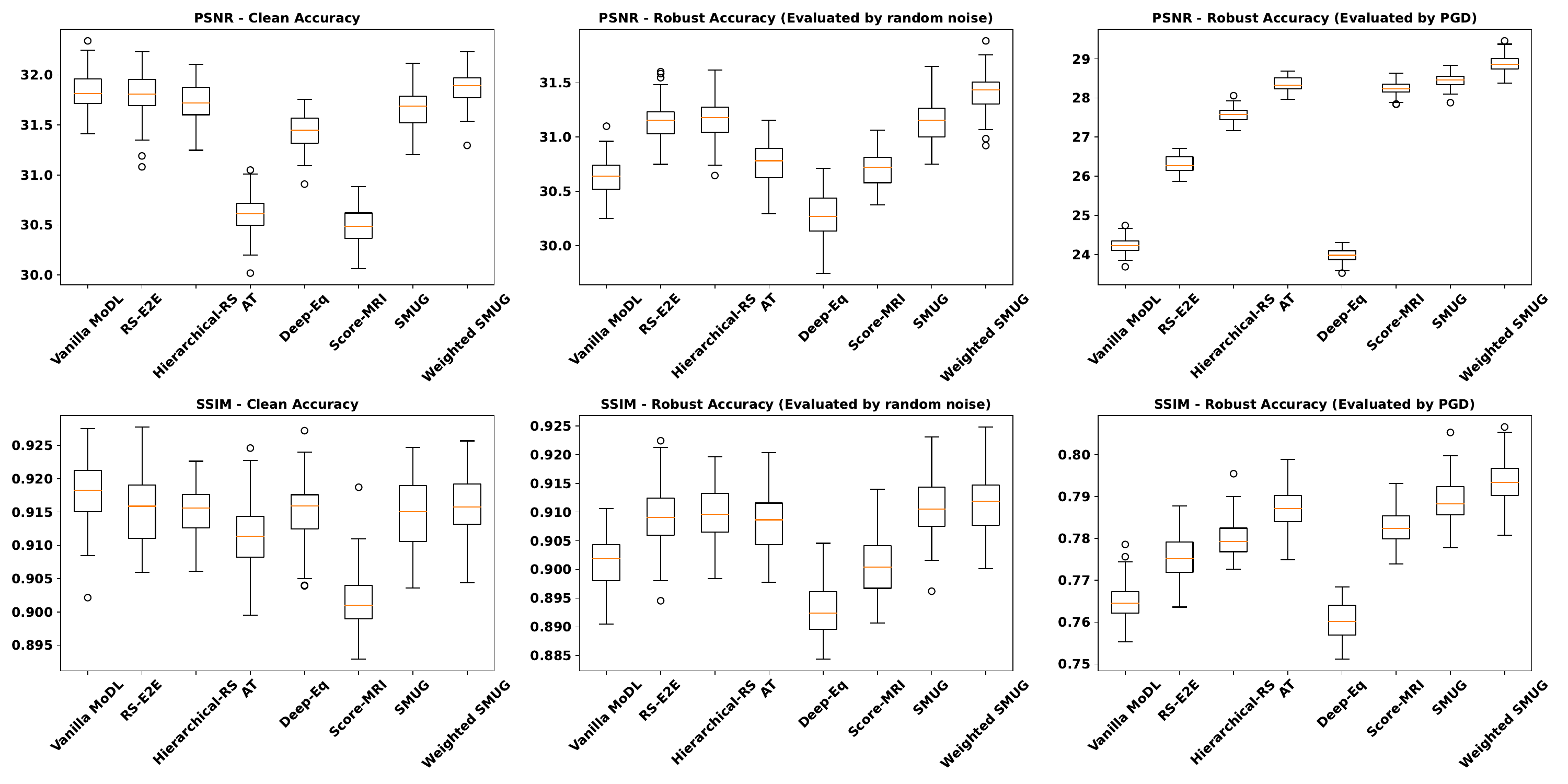}
    \vspace{-0.35cm}
    \caption{{ \textcolor{black}{Reconstruction accuracy box plots for the fastMRI \textbf{knee} 
 dataset with 4x Acceleration factor. The additive random Gaussian noise of the second column plots is obtained using a standard deviation of $0.01$. The worst-case additive noise of the third  column is obtained using the PGD  method with $\epsilon = 0.02$}. }}
    \label{fig: box plot brain}
    \vspace{-0.1 in}
\end{figure*}

\begin{figure*}[!t]
\centering
\begin{tabular}[b]{cccc}
    \textbf{Ground Truth}& 
    \textbf{Vanilla {\modl}}&
    \textbf{RS-E2E}&
    \textbf{\textsc{SMUG}}\\  
    \includegraphics[width=.12\linewidth,valign=t]{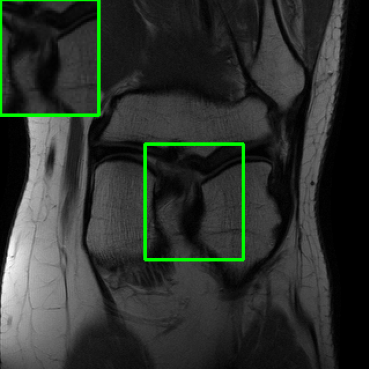}&
    \includegraphics[width=.12\linewidth,valign=t]{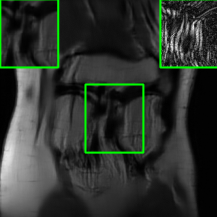}&
   \includegraphics[width=.12\linewidth,valign=t]{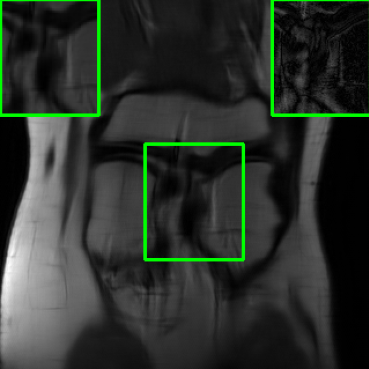} &
    \includegraphics[width=.12\linewidth,valign=t]{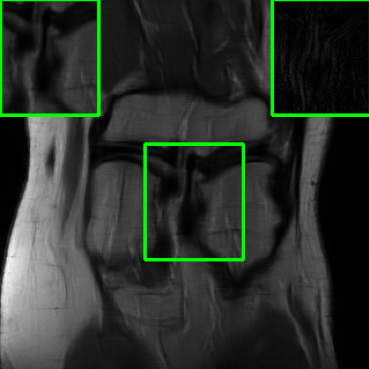}\\
    \scriptsize{PSNR = $\infty$ dB} & \scriptsize{PSNR = 23.41 dB} & \scriptsize{PSNR = 24.58 dB}& \scriptsize{\textbf{PSNR = 28.91 dB}}\\
    \textbf{AT}&
    \textbf{Score-MRI}&
    \textbf{\small Deep-Equilibrium}&
    \textbf{\small Weighted-SMUG}\\  
    \includegraphics[width=.12\linewidth,valign=t]{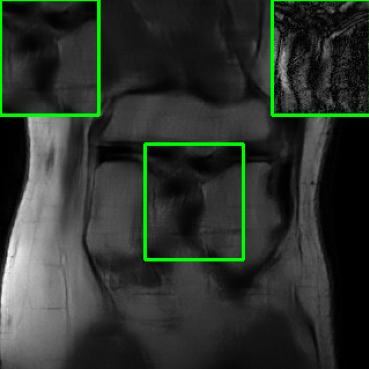} &
    \includegraphics[width=.12\linewidth,valign=t]{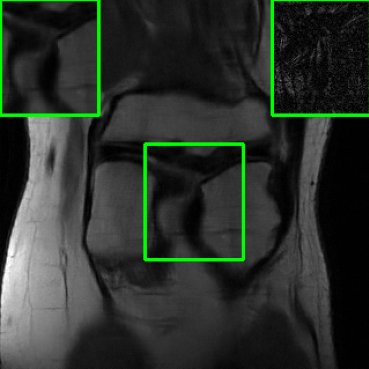} &
    \includegraphics[width=.12\linewidth,valign=t]{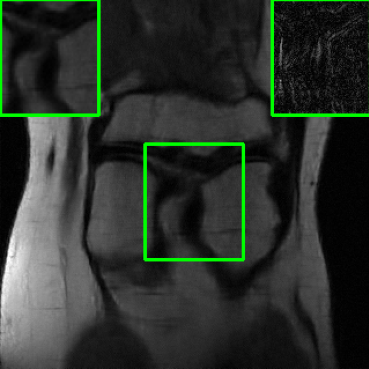} &
    \includegraphics[width=.12\linewidth,valign=t]{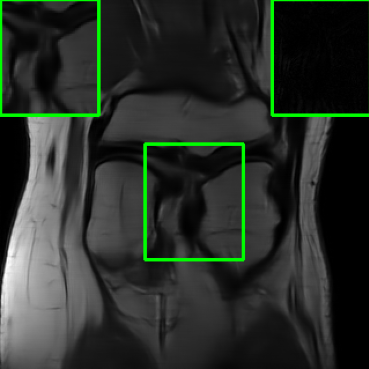}\\
        \scriptsize{PSNR = 28.67 dB} & \scriptsize{PSNR = 27.89 dB} & \scriptsize{PSNR = 24.13 dB}& \scriptsize{\textbf{PSNR = 29.41 dB}}\\
\end{tabular}
\caption{\textcolor{black}{Visualization of ground-truth and reconstructed images using different methods for 4x k-space undersampling, evaluated on PGD-generated worst-case inputs of perturbation strength $\epsilon = 0.02$.} \textcolor{black}{The reconstruction PSNRs are shown with the best values bolded.}}
\label{fig:denoised_imgs_zoomed_3000case_4x}
\vspace{-0.1 in}
\end{figure*}

\begin{figure*}
    \centering
    \includegraphics[width=0.88\textwidth]{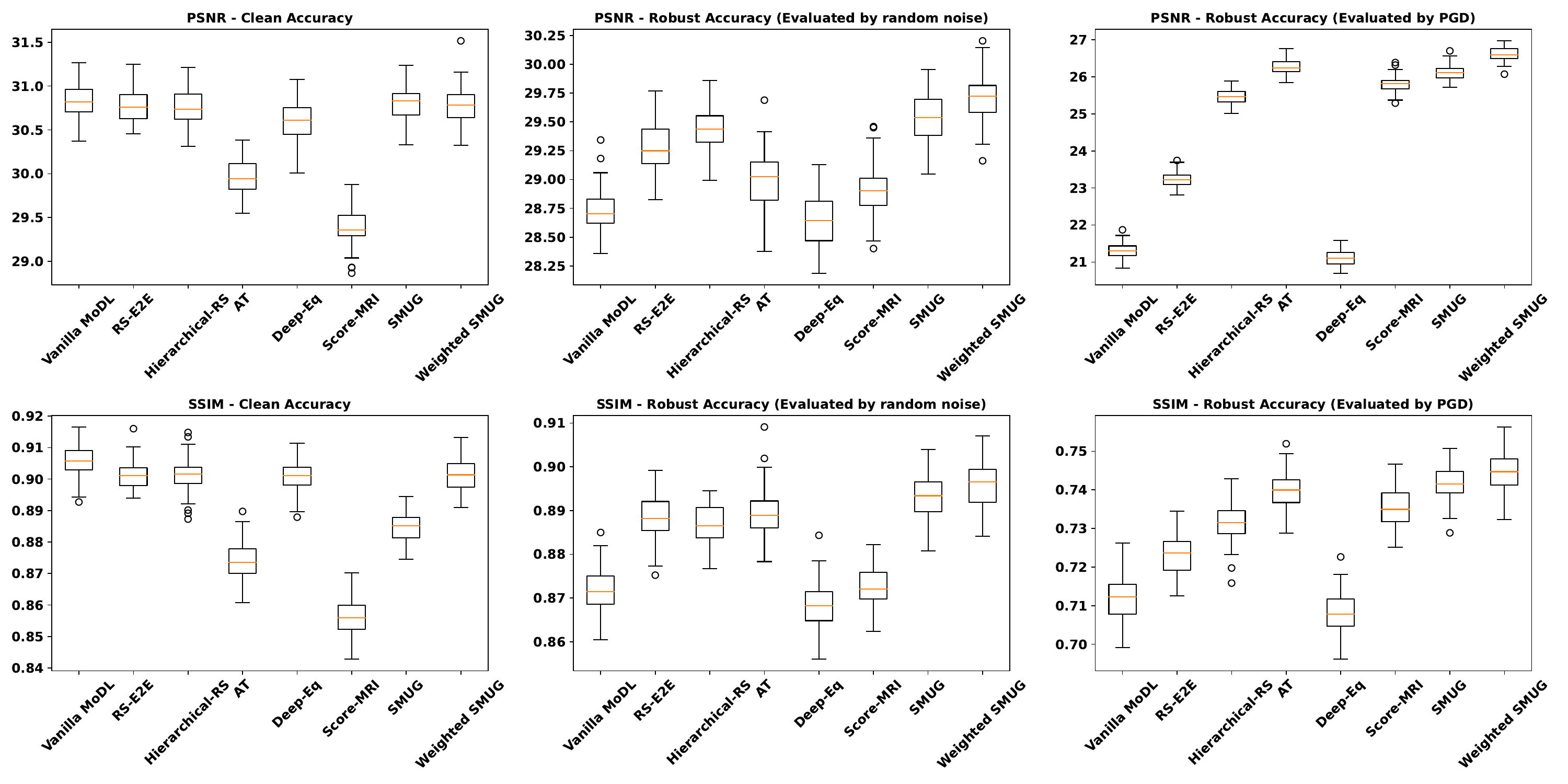}
    \vspace{-0.35cm}
    \caption{{ \textcolor{black}{Reconstruction accuracy box plots for the fastMRI \textbf{knee} 
 dataset with 8x Acceleration factor. The additive random Gaussian noise in the second column plots is obtained using a standard deviation of $0.01$. The worst-case additive noise in the third  column is obtained using the PGD  method with $\epsilon = 0.02$}.}}
    \label{fig: box plot brain 8x}
    \vspace{-0.15 in}
\end{figure*}

\begin{figure*}[!t]
\centering
\begin{tabular}[b]{cccc}
    \textbf{Ground Truth}& 
    \textbf{Vanilla {\modl}}&
    \textbf{RS-E2E}&
    \textbf{\textsc{SMUG}}\\  
    \includegraphics[width=.12\linewidth,valign=t]{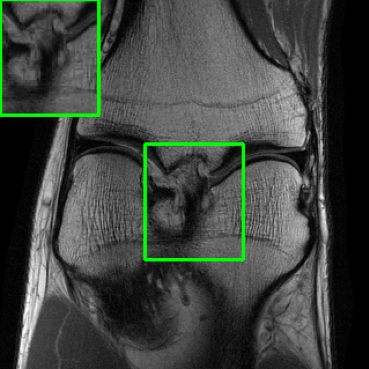}&
    \includegraphics[width=.12\linewidth,valign=t]{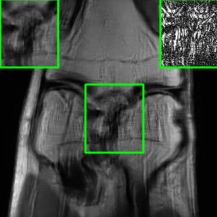}&
   \includegraphics[width=.12\linewidth,valign=t]{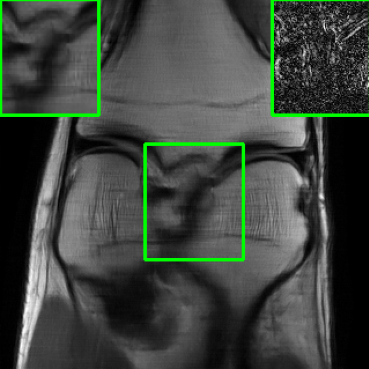} &
    \includegraphics[width=.12\linewidth,valign=t]{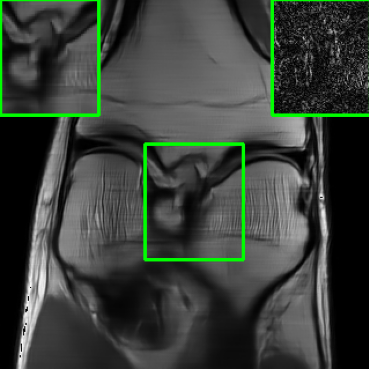}\\
    \scriptsize{PSNR = $\infty$ dB} & \scriptsize{PSNR = 21.48 dB} & \scriptsize{PSNR = 23.09 dB}& \scriptsize{\textbf{PSNR = 26.51 dB}}\\
    \textbf{AT}&
    \textbf{Score-MRI}&
    \textbf{\small Deep-Equilibrium}&
    \textbf{\small Weighted-SMUG}\\  
    \includegraphics[width=.12\linewidth,valign=t]{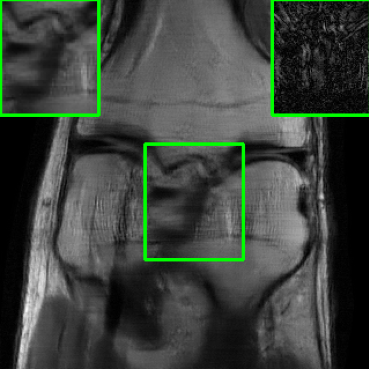} &
    \includegraphics[width=.12\linewidth,valign=t]{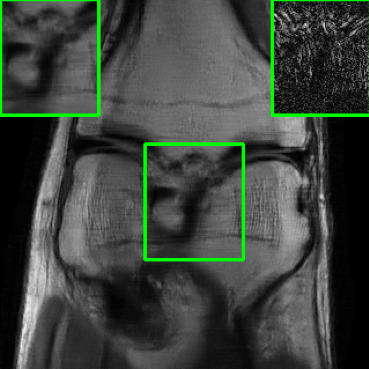} &
    \includegraphics[width=.12\linewidth,valign=t]{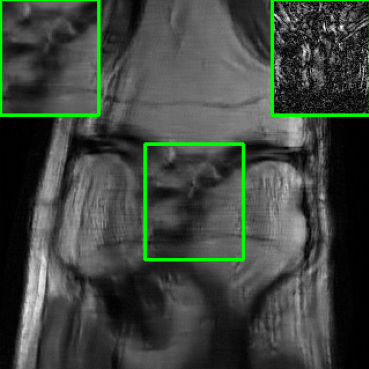} &
    \includegraphics[width=.12\linewidth,valign=t]{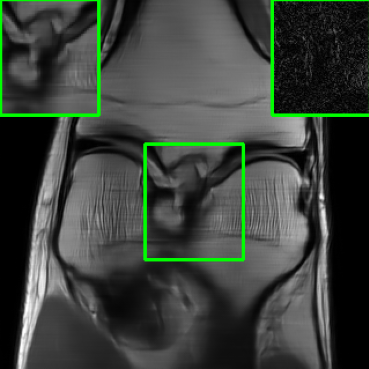}\\
        \scriptsize{PSNR = 26.34 dB} & \scriptsize{PSNR = 25.78 dB} & \scriptsize{PSNR = 21.32 dB}& \scriptsize{\textbf{PSNR = 26.89 dB}}\\
\end{tabular}

\caption{\textcolor{black}{Visualization of ground truth and reconstructed images using different methods for 8x k-space undersampling, evaluated on PGD-generated worst-case inputs of perturbation scaling $\epsilon = 0.02$.} \textcolor{black}{The reconstruction PSNRs are also shown with the best values bolded.}}
\label{fig:denoised_imgs_zoomed_3000case_8x}
\vspace{-0.1 in}
\end{figure*}

\subsection{Experimental Setup} 
\label{sec:experiment:setup}


\noindent \textbf{Models \& Sampling Masks:} For the MoDL architecture, we use the recent state-of-the-art denoising network Deep iterative Down Network, which consists of $3$ down-up blocks (DUBs) and 64 channels \cite{yu2019deep}. Additionally, for {\modl}, we use $N = 8$ unrolling steps with denoising regularization parameter $\lambda = 1$. The conjugate gradient method \cite{Aggarwal2019MoDL:Problems}, with a tolerance level of $10^{-6}$, is utilized to execute the DC block. 
\textcolor{black}{We used variable density 
Cartesian random undersampling masks in 
k-space, one for each undersampling factor that include a fully-sampled central k-space region and the remaining phase encode lines were sampled uniformly at random.}
The coil sensitivity maps for all scenarios were generated with the BART toolbox~\cite{tamir2016generalized}.
\textcolor{black}{Extension to the ISTA-Net model is discussed in Section~\ref{istanetresults}.}

\noindent \textbf{Baselines:} \textcolor{black}{We consider three robustification approaches: the RS-E2E method~\cite{jia2022on} presented in \eqref{eq: denoised smoothing mri}, Adversarial Training (AT)~\cite{jia2022robustness}, and the recent Hierarchical Randomized Smoothing~\cite{scholten2024hierarchical}. Furthermore, we consider other recent reconstruction models, specifically, the Deep Equilibrium (Deep-Eq) method~\cite{Deep-eq} and a leading diffusion-based MRI reconstruction model from~\cite{chung2022score}, which we denote as Score-MRI.}

\noindent \textbf{Datasets \& Training:} For our study, we execute two experimental cases. For the first case, we utilize the \texttt{fastMRI} knee dataset, with 32 scans for validation and 64 unseen scans/slices for testing. In the second case, we employ our method for the \texttt{fastMRI} brain dataset. 
We used $3000$ training scans in both cases.
\textcolor{black}{The k-space data are normalized so that the real and imaginary components are in the range $[-1, 1]$.}
We use a batch size of $2$ and $60$ training epochs. The experiments are run using two A5000 GPUs. The ADAM optimizer \cite{kingma2014adam} is utilized for training the network weights with momentum parameters of $(0.5, 0.999)$ and learning rate of $10^{-4}$. The stability parameter $\lambda_\ell$ in \eqref{eq: finetune_loss} (and \eqref{eq: finetune_loss_weighted}) is tuned so that the standard accuracy of the learned model is comparable to vanilla {\modl}.
For RS-E2E, we set the standard deviation of Gaussian noise to $\sigma = 0.01$, and use $10$ Monte Carlo samplings to implement the smoothing operation. Note that in our experiments, Gaussian noise and corruptions are added to real and imaginary parts of the data with the indicated $\sigma$. 

\textcolor{black}{For AT, we implemented a 30-step PGD procedure within its minimax formulation with $\epsilon = 0.02$. For Score MRI, we used 150 steps for the reverse diffusion process with the pre-trained model\footnote{\url{https://github.com/HJ-harry/score-MRI}}. 
We fine-tuned a pre-trained Deep-Eq model\footnote{\url{https://github.com/dgilton/deep_equilibrium_inverse}} with the same data as the proposed schemes.
Unless specified, training parameters were similar across the compared methods.}


\noindent \textbf{Testing:} We evaluate our methods on clean data (without additional perturbations), data with randomly injected noise, and data contaminated with worst-case additive perturbations. The worst-case disturbances allow us to see worst-case method sensitivity and
are generated by the $\ell_\infty$-norm based PGD scheme with 10 steps \cite{antun2020instabilities} corresponding to  $\left \| \boldsymbol \delta \right \|_{\infty} \leq \epsilon$, where $\epsilon$ is set nominally as the maximum underlying k-space real and imaginary part magnitude scaled by $0.05$. We will indicate the scaling for $\epsilon$ (e.g., $0.05$) in the results and plots that follow. The quality of reconstructed images is measured using  
peak signal-to-noise ratio (PSNR) and structure similarity index measure (SSIM)\cite{wang2004image}.
In addition to the worst-case perturbations and random noise, we evaluate the performance of our methods in the presence of additional instability sources such as (i) different undersampling rates, and (ii) different numbers of unrolling steps.

\begin{figure}[hbt!]
\vspace{-0.0in}
\centering
\setlength{\tabcolsep}{0.3cm}
\includegraphics[width=1.0\linewidth]{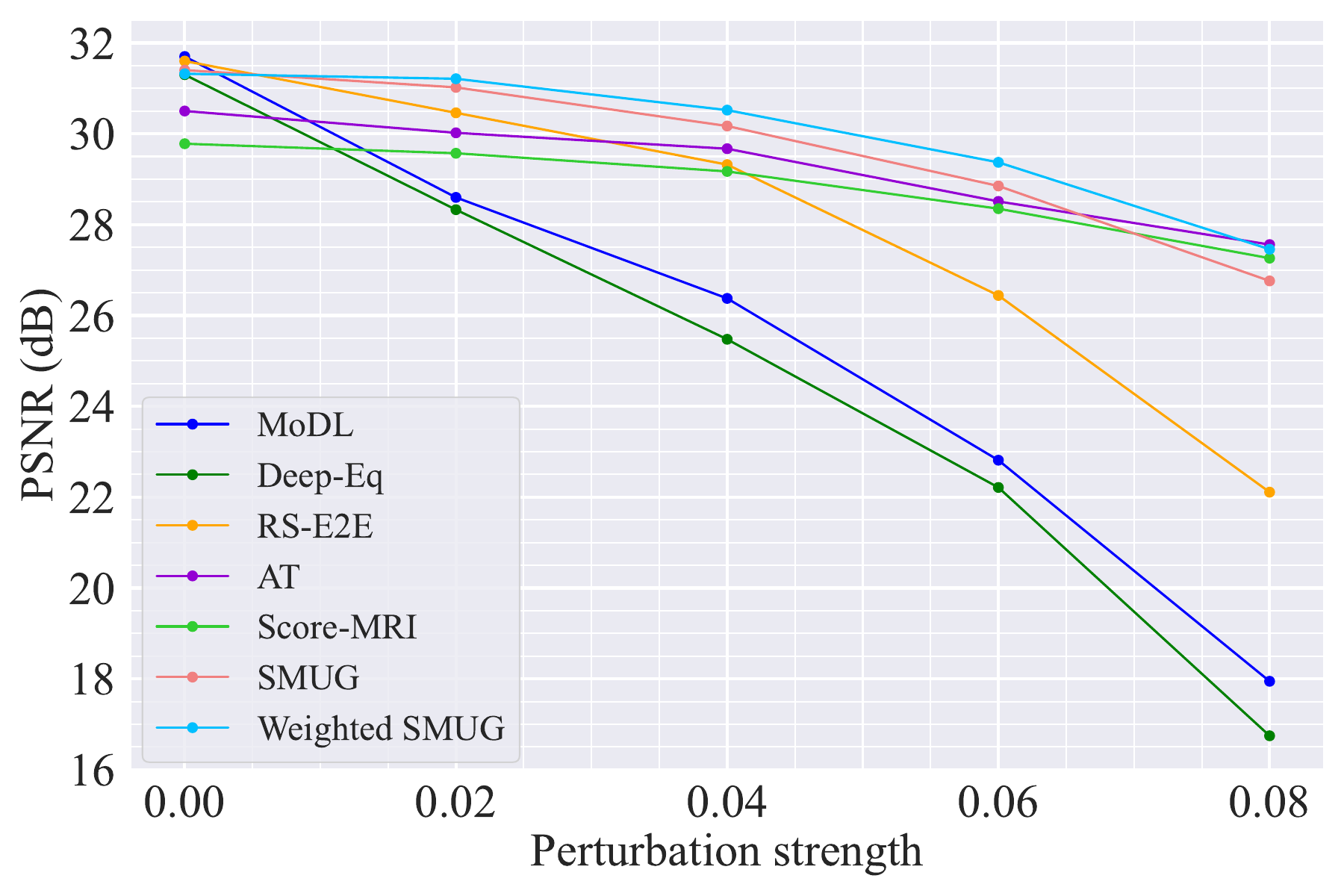}
\vspace{-0.35cm}
  \caption{\textcolor{black}{PSNR  of baseline methods and the proposed method versus perturbation strength (i.e., scaling) $\epsilon$ used in  PGD-generated worst-case examples at testing time with 4x k-space undersampling. $\epsilon = 0$ corresponds to clean accuracy. 
    }}
\label{fig:archi_PSNR}
\vspace{-0.2 in}
\end{figure}

\begin{figure}[hbt!]
\centering
\setlength{\tabcolsep}{0.3cm}
\includegraphics[width=1.0\linewidth]{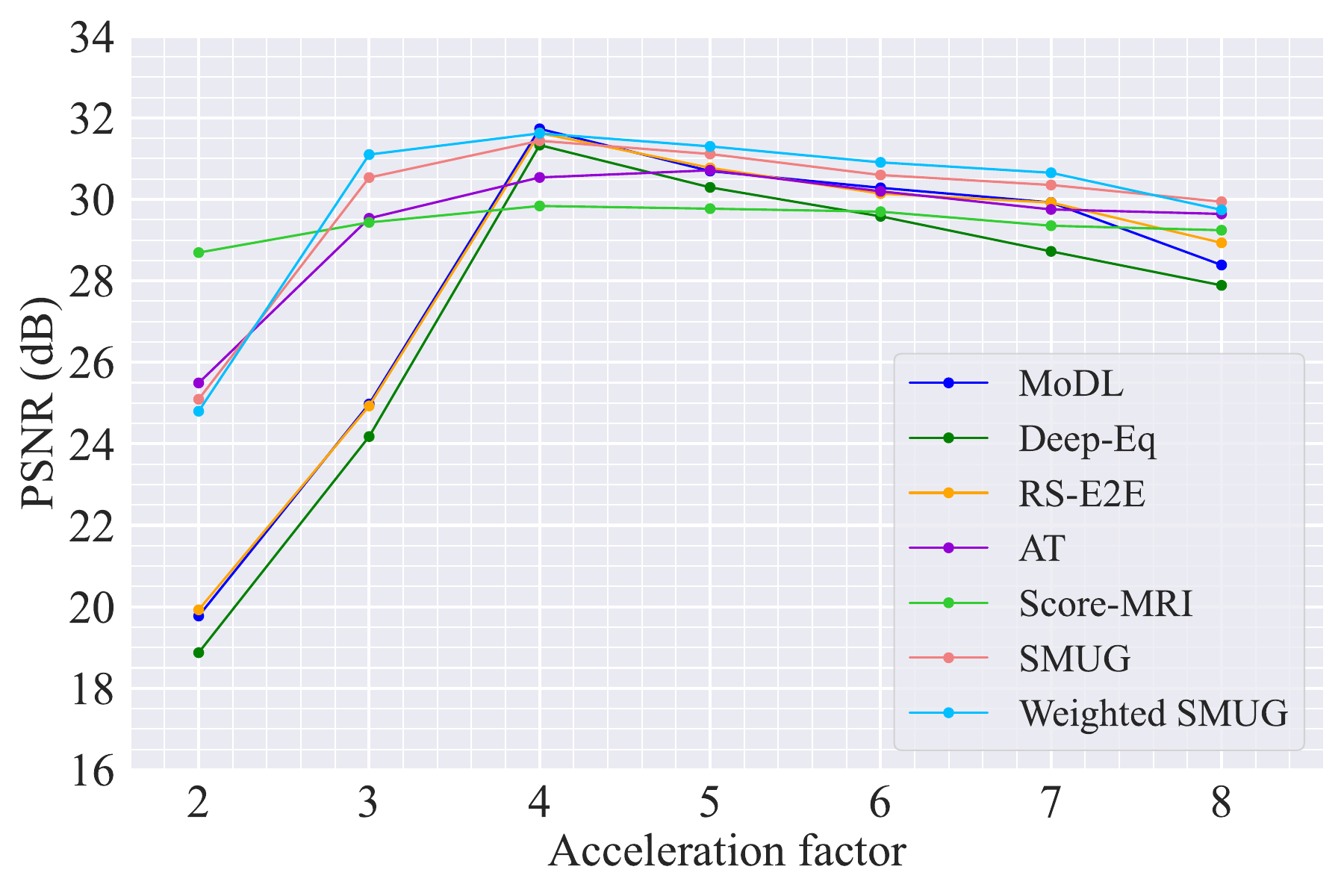}
\vspace{-0.4cm}
    \caption{\textcolor{black}{PSNR results for different MRI reconstruction methods versus
    different measurement sampling rates (models trained at $4\times$ acceleration).}}
    \label{fig: robustness}
    \vspace{-0.14in}
\end{figure}

\subsection{Robustness Results}
\textcolor{black}{\noindent \textbf{Results for the FastMRI Brain Dataset:} we present the robustness results of the proposed approaches w.r.t. additive noise. In particular, the evaluation is conducted on the clean, noisy (with added Gaussian noise), and worst-case perturbed (using PGD for each method) measurements. \textbf{Fig.~\ref{fig: box plot knee}} presents testing set PSNR and SSIM values as box plots for different smoothing architectures, 
along with vanilla {\modl} and the other baselines using the brain dataset. The clean accuracies of Weighted {\us} and {\us} are similar to vanilla MoDL indicating a good clean accuracy vs. robustness trade-off.
As indicated by the PSNR and SSIM values, we observe that weighted {\us}, on average, outperforms all other baselines in robust accuracy (the second and third set of box plots of the two rows in \textbf{Fig.~\ref{fig: box plot knee}}). This observation is consistent with the visualization of reconstructed images for the brain dataset in \textbf{Fig.\,\ref{fig:denoised_imgs_zoomed_300case}}. 
We note that weighted {\us} requires longer time for training, which represents a trade-off. 
When comparing to AT, we observe that AT is comparable to {\us} in the case of robust (or worst-case noise) accuracy. However, the drop in clean accuracy (without perturbations) for AT is significantly larger than for {\us}. Furthermore, AT takes a much longer training time as it requires to solve an optimization problem (PGD) for every training data sample at every iteration to obtain the worst-case perturbations. 
Furthermore, we observe that its effectiveness is degraded for other perturbations including random noise as well as modified sampling rates 
shown in the next subsection. Importantly, the proposed SMUG and Weighted SMUG are not trained to be robust to any specific perturbations or instabilities, but are nevertheless effective for several scenarios.}

\textcolor{black}{In comparison to the diffusion based Score-MRI, the proposed methods perform better in terms of both
clean accuracy and random noise accuracy. 
Although for worst-case perturbations, the PSNR values of Score-MRI are only slightly worse than {\us}, \textcolor{black}{it is important to note that not only the training of diffusion-based models takes longer than our method, but also the inference time is longer as Score-MRI requires to perform nearly 150 sampling steps to process one scan and takes nearly 5 minutes with a single RTX5000 GPU, whereas our method takes only about 25 seconds per scan.} The SMUG schemes also substantially outperform the deep equilibrium model in the presence of perturbations.}

\noindent \textbf{Results for the FastMRI Knee Dataset:} \textcolor{black}{In \textbf{Fig~\ref{fig: box plot brain}} and \textbf{Fig~\ref{fig: box plot brain 8x}}, we report PSNR and SSIM results of different methods at two sampling acceleration factors for the knee dataset. Therein, we observe quite similar outcomes to those reported in \textbf{Fig~\ref{fig: box plot knee}}. }
\textcolor{black}{\textbf{Figs.~\ref{fig:denoised_imgs_zoomed_3000case_4x}} and~\textbf{\ref{fig:denoised_imgs_zoomed_3000case_8x}} show reconstructed images by different methods for knee scans at 4x and 8x undersampling, respectively. 
We observe that {\us} and Weighted {\us} show fewer artifacts, sharper features, and fewer errors when compared to Vanilla MoDL and other baselines in the presence of the worst-case perturbations.}

\noindent \textbf{Results on Adversarial Perturbation Strength:} In \textcolor{black}{\textbf{Fig.~\ref{fig:archi_PSNR}} presents average PSNR results over the test dataset for the considered models under different levels of worst-case perturbations (\textit{i.e.}, attack strength $\epsilon$). We used the knee dataset for this experiment. We observe that SMUG and weighted SMUG outperform RS-E2E, vanilla {\modl}, and Deep-Eq across all perturbation strengths. When compared to Score-MRI and AT, our proposed methods consistently maintain higher PSNR values for moderate to large perturbations (less than $\epsilon=0.08$). For instance, when $\epsilon = 0.02$, weighted {\us} reports more than 1 dB improvement over AT and Score-MRI.}

\noindent \textbf{Impact of the Undersampling Rate Disparities:} During training, a k-space undersampling or acceleration factor of 4x is used for our methods and the considered baselines. At testing time, we evaluate performance (in terms of PSNR) with acceleration factors ranging from 2x to 8x. The results are presented in \textbf{Fig.\,\ref{fig: robustness}}. It is clear that when the acceleration factor during testing matches that of the training phase (4x), all methods achieve their highest PSNR results.  \textcolor{black}{Conversely, performance generally declines when the acceleration factors differ. For acceleration factors 3x to 8x (ignoring 4x where models were trained), we observe that our methods outperform all the considered baselines. For the 2x case, our methods report higher PSNR values compared to RS-E2E, vanilla {\modl}, and Deep-Eq and slightly underperform AT, while Score-MRI shows more resilience at 2x.}

\begin{figure}[hbt!]
\centering
\setlength{\tabcolsep}{0.3cm}
\includegraphics[width=1.0\linewidth] {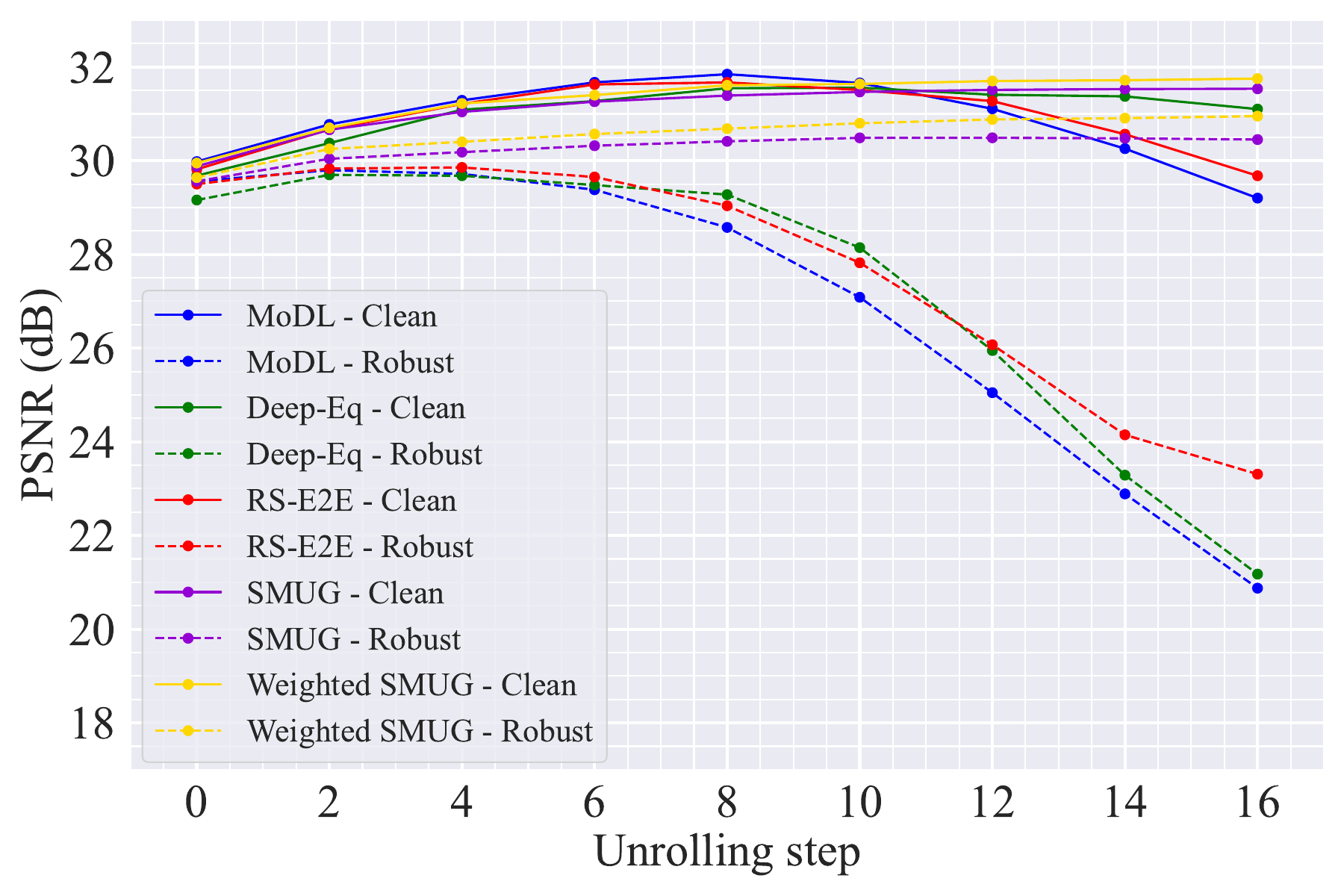}
\vspace{-0.3cm}
    \caption{
    \textcolor{black}{PSNR results for different MRI reconstruction methods at 4x k-space undersampling versus number of
    unrolling steps ($8$ steps used in training). ``Clean" and "Robust" denote the cases without and with added worst-case (for each method) measurement perturbations.}}
    \label{fig: robustness_unrolling}
    \vspace{-0.1in}
\end{figure}

\noindent \textbf{Results for the Unrolling Steps Disparities:} Here, we study the performance of varying unrolling steps. More specifically, during training, we utilize 8 unrolling steps to train our methods and the baselines. At testing time, we report the results of utilizing 1 to 16 unrolling steps. \textcolor{black}{The PSNR results of all the considered cases are given in \textbf{Fig.\, \ref{fig: robustness_unrolling}}. The results show that both {\us} and Weighted {\us} maintain performance comparable to the Deep Equilibrium model. 
Furthermore, when using different unrolling steps and faced with additive measurement perturbations, the SMUG methods' PSNR values are stable and close to the unperturbed case (indicating robustness), whereas the other methods see more drastic drop in performance.
This behavior for SMUG also agrees with the theoretical bounds in Section~\ref{method}.}

Although we do not intentionally design our method to mitigate {\modl}'s instabilities against different sampling rates and unrolling steps, the SMUG approaches nevertheless provide improved PSNRs over other baselines. 
This indicates broader value for the robustification strategies incorporated in our schemes.

\subsection{Behavior of SMUG and Weighted SMUG} 

\noindent \textbf{Effect for the Ustab Loss:} We conduct additional studies on the unrolled stability loss in our scheme to
show the importance of integrating target image denoising into {\us}'s training pipeline in \eqref{eq: unrolling loss}. \textbf{Fig.\,\ref{fig: unrolling_loss}} presents PSNR values versus perturbation strength/scaling ($\epsilon$) when using different alternatives to $\mathcal D_{\btheta} (\mathbf t)$ in~\eqref{eq: unrolling loss}, including 
$\mathbf t$ (the original target image), $\mathcal D_{\btheta}(\mathbf x_n)$ (denoised output of each unrolling step), and variants when using the fixed, vanilla {\modl}'s denoiser $\mathcal D_{\btheta_\text{\modl}}$ instead.
As we can see, the performance of {\us} varies when the UStab loss \eqref{eq: unrolling loss} is configured differently. The proposed $\mathcal D_{\btheta}(\mathbf t)$ outperforms other baselines. \textcolor{black}{A possible reason is that it infuses supervision of target images in an adaptive, denoising-friendly manner, \textit{i.e.}, taking the influence of $\mathcal D_{\btheta}$ into consideration. The configuration involving $\mathcal D_{\btheta}(\x_n)$ performs closest to using \(\mathcal D(t)\) in Fig~\ref{fig: unrolling_loss}. This indicates that a loss such as 
\(
\|\mathcal D_{\btheta}(\x_n + \eta) -\mathcal D_{\btheta}(\x_n)\|
\)
better guards the denoiser behavior with respect to noise perturbations compared to directly fitting the target $\mathbf t$. We conjecture the reason using \(D_{\btheta}(\mathbf t)\) is even more robust is because it enables the denoiser to mimic auto-encoding the target (note that the original regularization in the MoDL scheme is to do auto-encoding).}

\begin{figure}[hbt!]
\vspace{-0.01in}
\centering
\hspace{-0.05 in}
\includegraphics[width=0.5\textwidth]{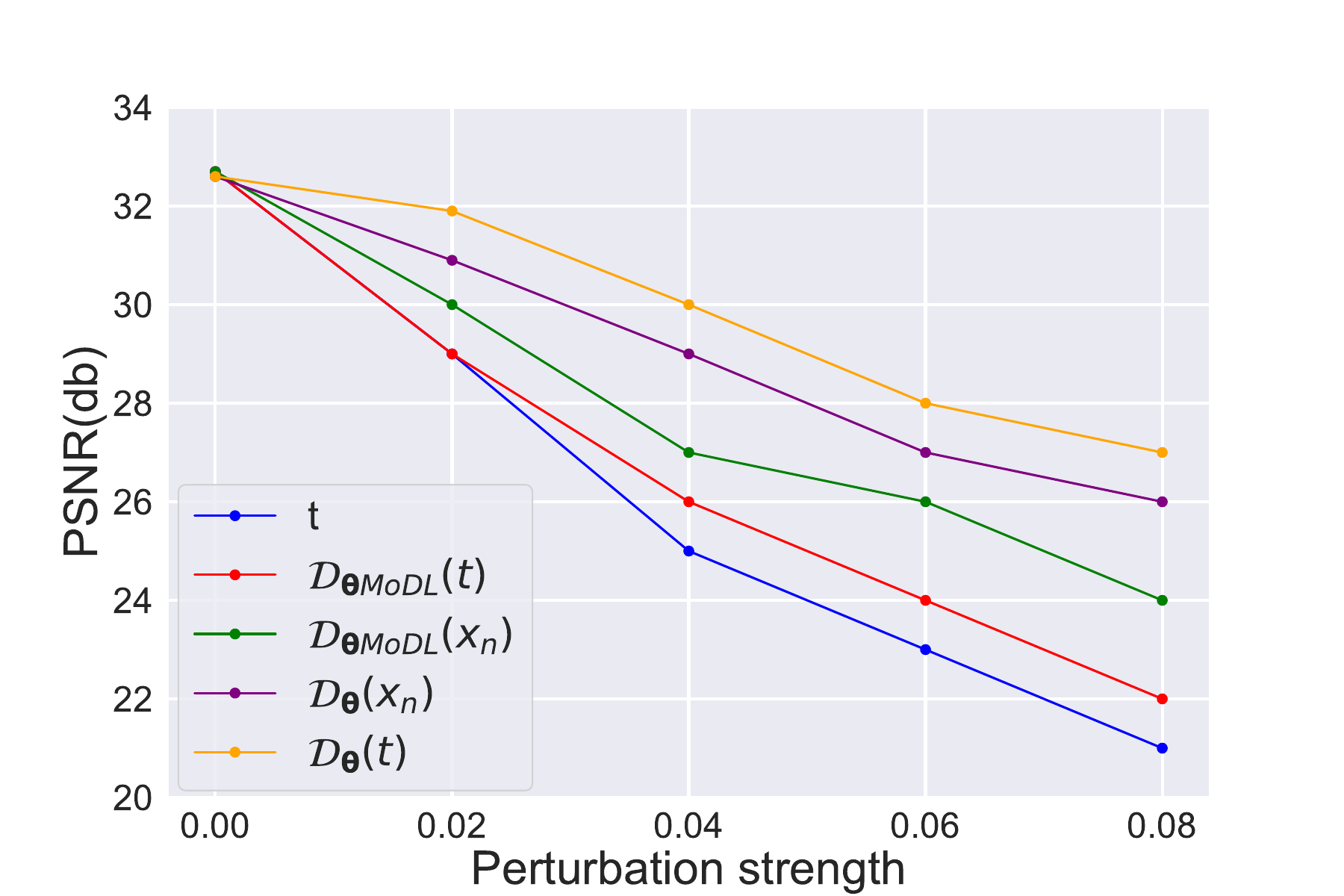}
\vspace{-0.1in}
    \caption{{PSNR vs. worst-case perturbation strength ($\epsilon)$ for {\us} for different configurations of   UStab loss \eqref{eq: unrolling loss}. 
    }}
    \label{fig: unrolling_loss}
\vspace{-0.1in}
\end{figure}
\begin{figure}[htb]
\centering
\begin{tabular}{cc}
  \hspace*{-0.15in}  
  \includegraphics[width=0.25\textwidth]{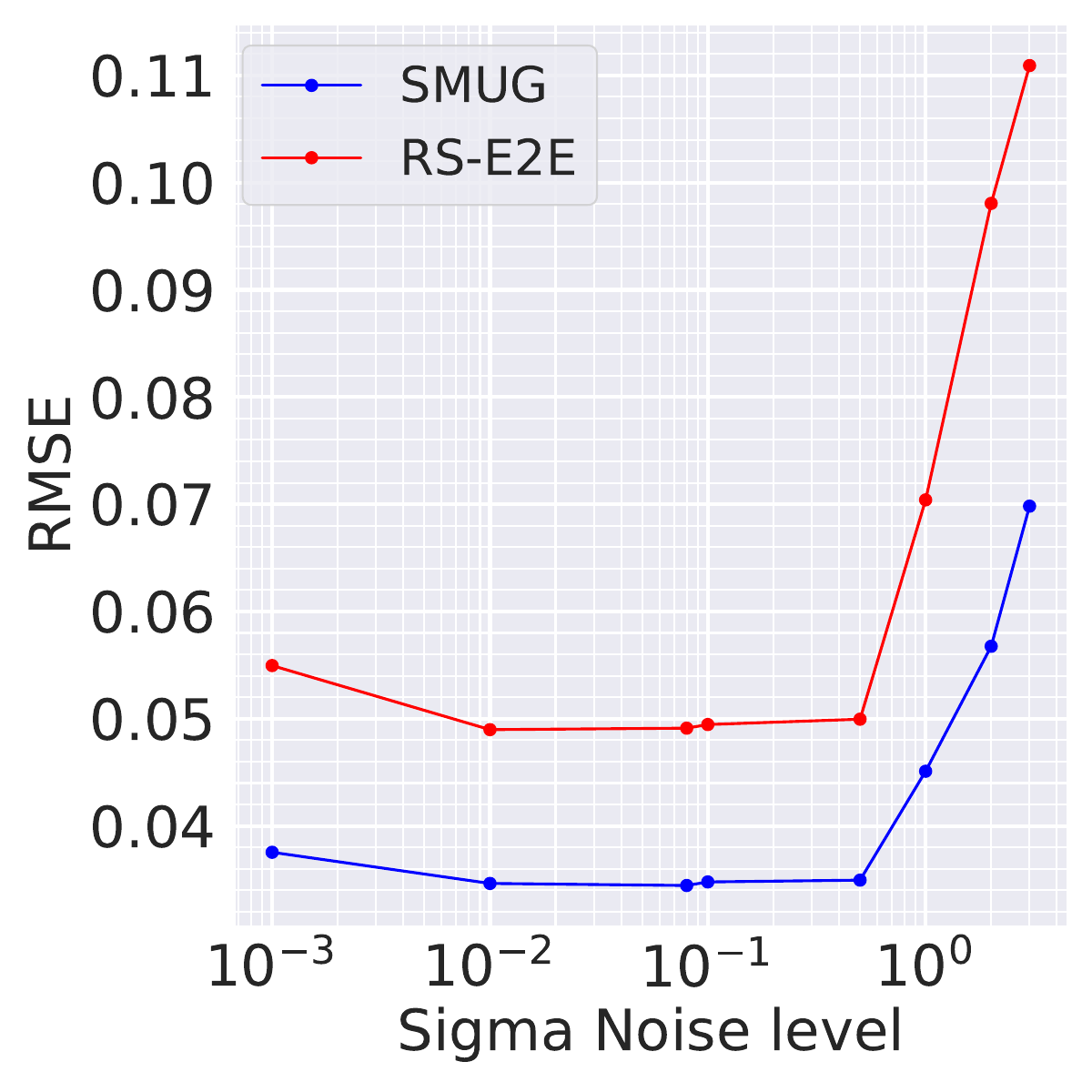}   &    \hspace*{-5mm}\includegraphics[width=0.25\textwidth]{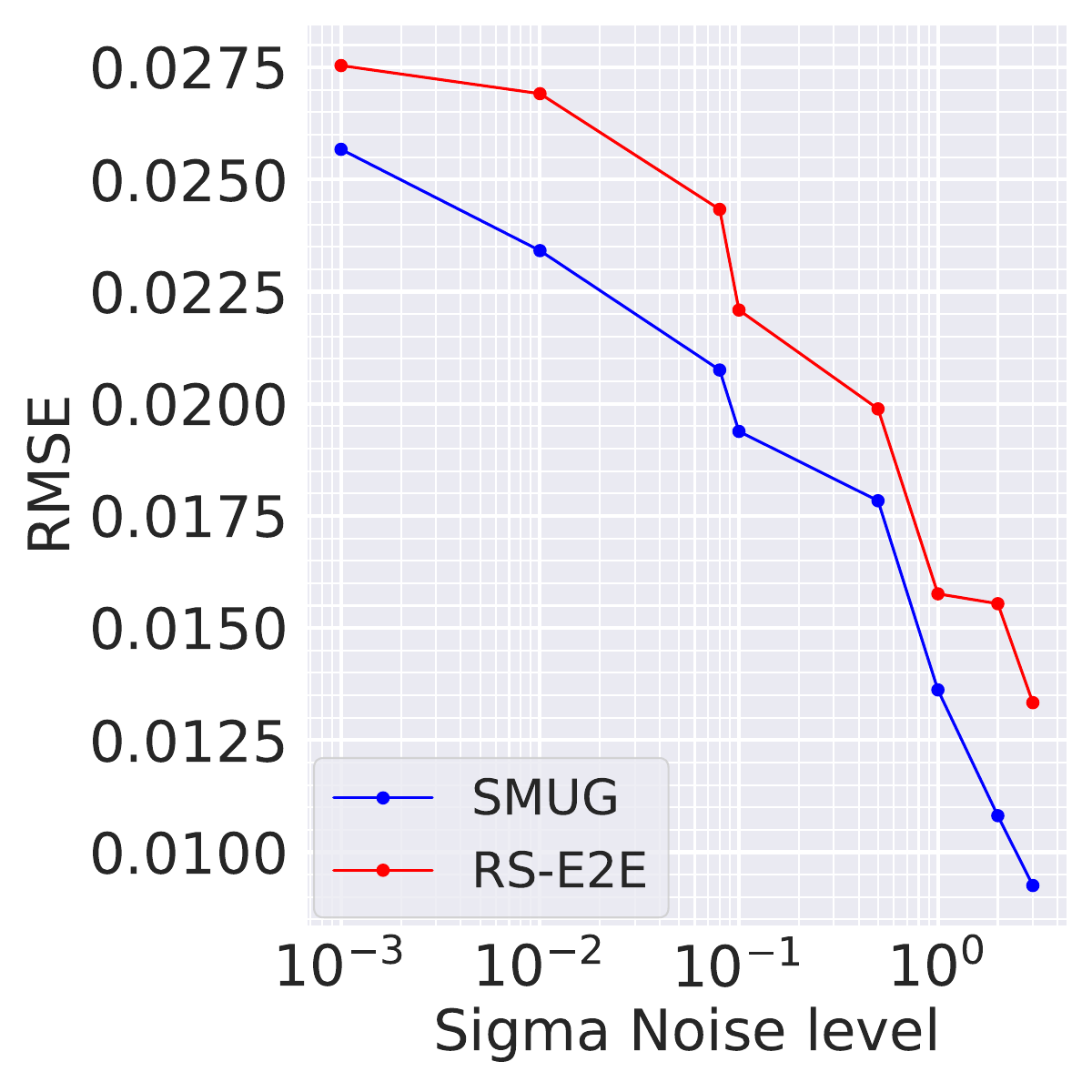}
  \vspace{-0.1cm}
\end{tabular}
    \caption{{Left: Norm of difference between SMUG and RS-E2E reconstructions and the ground truth for different choices of $\sigma$ in the smoothing process. A worst-case PGD perturbation $\boldsymbol\delta$ computed at $\epsilon=0.01$ was added to the measurements in all cases.
    Right: Robustness error for SMUG and RS-E2E at various $\sigma$, i.e., norm of difference between output with the perturbation $\boldsymbol\delta$ and without it.
    }}
    \label{fig: reference_PSNR}
  \vspace*{-0.2cm}
\end{figure}
\begin{figure}[!t]
\vspace{-0.01in}
\centering
\includegraphics[width=0.45\textwidth]{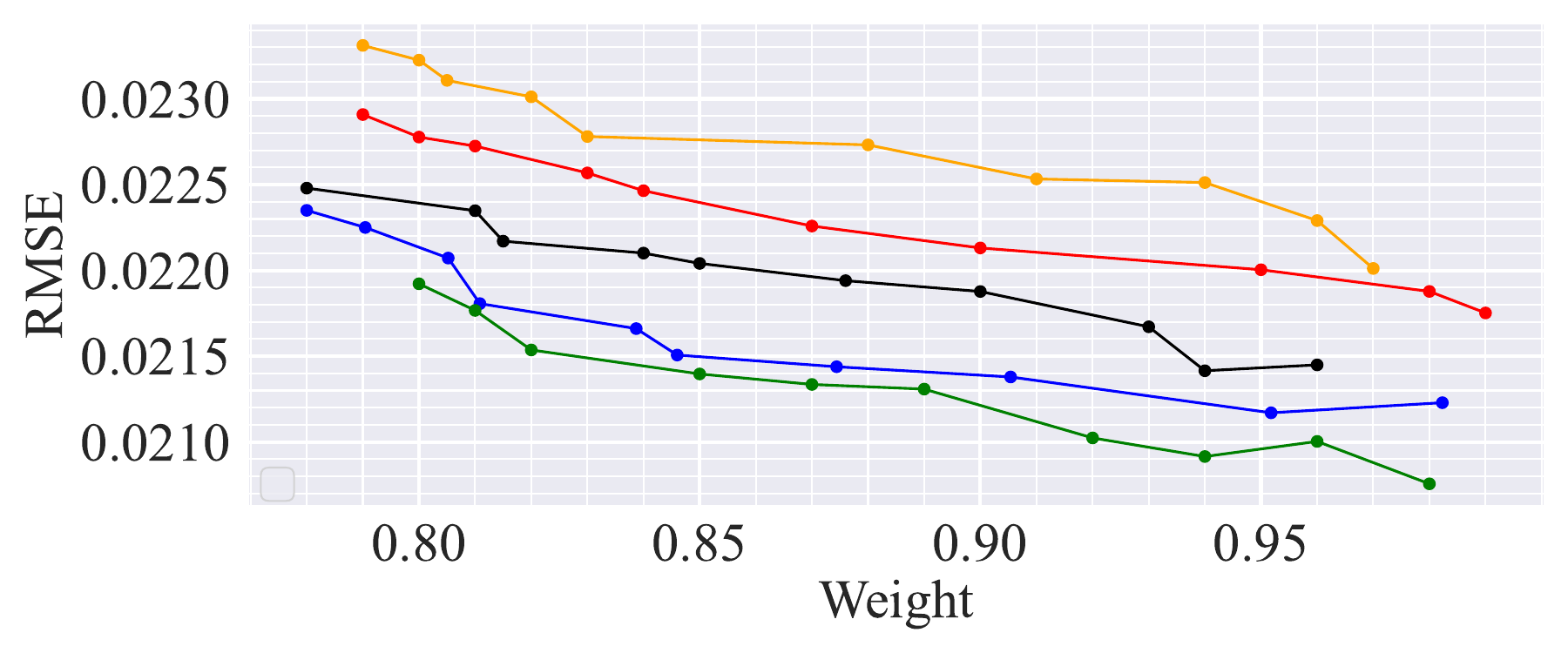}
\caption{Weights predicted by the weight encoder network in Weighted SMUG (from final layer of unrolling) plotted against root mean squared error (RMSE) of the corresponding denoised images for $5$ randomly selected scans (with 4x undersampling).}
\label{weight check}
\vspace{-0.2in}
\end{figure}

\noindent \textbf{Impact of the Noise Smoothing:} To comprehensively assess the influence of the introduced noise during smoothing, denoted as $\boldsymbol{\eta}$,
on the efficacy of the suggested approaches, we undertake an experiment involving varying noise standard deviations $\sigma$. 
The outcomes, documented in terms of RMSE, are showcased in $\textbf{Fig.} \ref{fig: reference_PSNR}$. 
The accuracy (reconstruction quality w.r.t. ground truth) and robustness error (error between with and without measurement perturbation cases) are shown for both {\us} and RS-E2E.
We notice a notable trend: as the noise level $\sigma$ increases, the accuracy for both methods improves before beginning to degrade. 
Importantly, {\us} consistently outperforms end-to-end smoothing.
Furthermore, the robustness error continually drops as $\sigma$ increases (corroborating with our analysis/bound in Section~\ref{method}), with more rapid decrease for {\us}.

\noindent \textbf{Empirical Analysis of the behavior of Weighted {\us}:} In \textbf{Fig. \ref{weight check}}, we analyze the behavior of the Weighted {\us} algorithm.
We delve into the nuances of weighted smoothing, which can assign different weights to different images during the smoothing process. The aim is to gauge how the superior performance of Weighted {\us} arises from the variations in learned weights.
Our findings indicate that among the $10$ Monte Carlo samplings implemented for the smoothing operation, those with lower denoising RMSE when compared to the ground truth images generally receive higher weights.

\vspace{0.1in}
\noindent \textbf{Computational Cost Analysis for SMUG and Weighted SMUG:} 
\textcolor{black}{Here, we do a computational cost analysis of the different smoothing-based methods to shed light on trade-offs.
At inference time, the proposed schemes involve randomized smoothing based denoising with a neural network and data consistency operations to enforce measurement priors. 
Let us assume the unrolled MoDL architecture. If we assume a denoising neural network with width \(M\) and depth \(L\), the cost for making a forward pass through it is \( \mathcal{O}(L M^2) \). In multi-coil MRI reconstruction, the conjugate gradient (CG) method iteratively solves the linear system~\eqref{modleqn1} involving the forward operator \(\mathbf{A}\), which consists of coil sensitivity weighting and an undersampled Fourier transform. Each iteration of CG involves applying \(\mathbf{A}\) and its adjoint \(\mathbf{A}^H\), requiring \( \mathcal{O}(N_c  n \log(n)) \) operations each, where \(N_c\) denotes the number of MR receiver coils, and \(n\) is the number of image pixels or voxels. The total complexity depends on the number of iterations in CG~\cite{greenbaum1997iterative}, typically 
\( k = 0.5\,\sqrt{\kappa(\mathbf{W})}\,\log\Bigl(\|\mathbf{e}_0\|_\mathbf{W}\,\epsilon^{-1}\Bigr) \), 
where \(\kappa(\cdot)\) denotes the condition number, the error tolerance achieved is \(2\epsilon > 0\), \(\mathbf{W} = \mathbf{A}^H \mathbf{A} + \lambda \mathbf{I}\), and \(\mathbf{e}_0\) denotes the initial iterate's error~\cite{greenbaum1997iterative}. The overall cost for the CG step is \( \mathcal{O}(k  N_c  n \log(n)) \). In practice, we observed only a few CG steps suffice. 
}

\textcolor{black}{For both end-to-end randomized Smoothing (RS-E2E) and SMUG, the overall computational cost is dependent on \(T\), the number of random noise samples used in the smoothing process. The computational cost for RS-E2E scales as \( \mathcal{O}(N  T  L  M^2) \) and \( \mathcal{O}(N  T  k  N_c  n \log(n)) \) for all the denoiser and CG steps (we assume a fixed $k$ for standard CG for simplicity), where \(N\) is the number of unrolling steps or iterations. 
The averaging of reconstructions (with different random noise perturbations) at the output of RS-E2E involves only \( \mathcal{O}(T  n) \) operations. 
On the other hand, for the proposed SMUG, the smoothing is performed in every step of denoising with the neural network. The total computational costs at inference time for the network forward passes, smoothing/averaging step, and CG step over \(N\) unrollings are \( \mathcal{O}(N  T  L  M^2) \), \( \mathcal{O}(N T  n) \), and \( \mathcal{O}(N  k  N_c  n \log(n)) \), respectively. 
Since the CG step usually takes longer than denoising and smoothing in practice, the cost for SMUG could be lower than RS-E2E since it does not apply CG for multiple noise-perturbed inputs. The costs for Weighted SMUG scales similarly as SMUG except for one more forward pass of noise-perturbed images through the weight prediction network.}

\textcolor{black}{Since the key difference between the typical unrolled network and SMUG is the iteration-wise randomized smoothing operation, we also performed an ablation study to show the effect of the number of noise samples used in the smoothing on overall image quality and runtime. \textbf{Fig.}~\ref{fig: noisy sample effect} demonstrates a clear trend, where an increase in the number of noise samples leads to a corresponding rise in reconstruction time, while the PSNR saturates after some point (by $10-12$ noise samples used for smoothing). This observation highlights the trade-off between utilizing more noise realizations for potentially improved reconstruction quality and the increased computational cost incurred at test time.}
\begin{figure}[!t]
\vspace{-0.01in}
\centering
\includegraphics[width=0.45\textwidth]{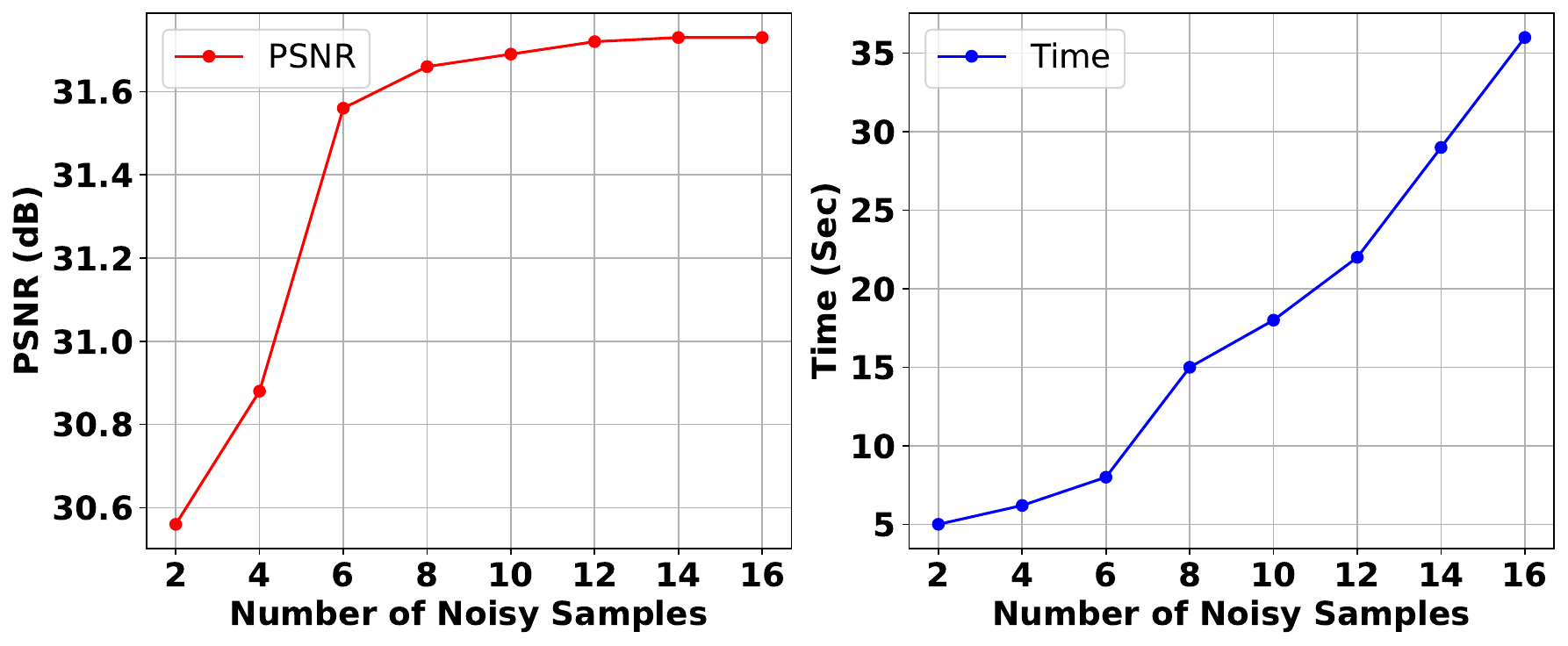}
\caption{\textcolor{black}{Runtime and PSNR trade-offs when varying the number of noisy samples during smoothing operations in SMUG. The results were averaged over $10$ randomly selected scans (with 4x undersampling).}}
\label{fig: noisy sample effect}
\vspace{-0.2in}
\end{figure}

\vspace{-0.1 in}
\subsection{Integrability of SMUG and Weighted SMUG in Other Unrolled Networks}  
\label{istanetresults}  

\textcolor{black}{In our concluding study, we explore whether our robustification methods maintain their effectiveness when applied to alternative unrolling techniques, specifically ISTA-Net~\cite{zhang2018istanet} and E2E-VarNet~\cite{sriram2020end}. While our experiments demonstrate promising results, we want to clarify that we do not claim SMUG or Weighted SMUG to be universally applicable for all unrolled networks. Instead, our goal is to establish its adaptability and effectiveness when integrated with some well-known unrolling-based architectures to further validate the robustness and generalizability of our methods.}  


\noindent \textbf{Applying Our Method to ISTA-Net:}  
For ISTA-Net, we adopted the default network architecture, utilizing the ADAM optimizer with a learning rate of $10^{-4}$. The network was configured with nine phases (unrolling iterations) and trained on the fastMRI knee dataset, which comprises 3,000 scans at a 4× undersampling rate. The training was conducted over 100 epochs to ensure adequate convergence. Consistent with our prior experimental setup, we used 64 scans for testing. All other training configurations for the vanilla ISTA-Net were set to their default values\footnote{\url{https://github.com/jianzhangcs/ISTA-Net-PyTorch}}, while the settings for the RS-E2E version, as well as the SMUG and Weighted SMUG variants of ISTA-Net, were aligned with those used in the MoDL experiments to facilitate a fair comparison.  

Figure~\ref{fig:box_plot_knee_ISTA_net} presents the performance evaluation, illustrating that both SMUG and Weighted SMUG versions of ISTA-Net achieve clean accuracy results comparable to the standard ISTA-Net. However, the key advantage of our method becomes evident in more challenging scenarios. Under conditions of random noise perturbation (Gaussian noise with $\sigma=0.01$) and adversarial interference from a PGD attack (30 steps with $\epsilon = 0.02$), our method demonstrates superior robustness. Specifically, both SMUG and Weighted SMUG outperform the original ISTA-Net as well as the RS-E2E variant, exhibiting improved resilience against noise and adversarial perturbations. These findings closely mirror the patterns observed when unrolling smoothing was applied to the MoDL network, reinforcing the efficacy of our approach across different architectures.  

\textcolor{black}{\noindent \textbf{Applying Our Method to E2E-VarNet:}  
To further assess the integrability of our method, we applied SMUG and Weighted SMUG to E2E-VarNet. In this case, we utilized the default architecture, consisting of 12 cascades (iterations of network refinement steps or unrolling steps). The network was optimized using ADAM with a learning rate of $3 \times 10^{-4}$. Other than this adjustment, most of the training settings for RS-E2E, as well as the SMUG and Weighted SMUG variants, remained consistent with those used for ISTA-Net to maintain comparability between experiments.}

\textcolor{black}{Our results reveal a strikingly similar trend to that observed in the ISTA-Net experiments. Specifically, while the clean performance of the SMUG and more so Weighted SMUG versions of E2E-VarNet remain on par with vanilla E2E-VarNet, their robustness under noisy and adversarial conditions significantly surpasses that of the standard model. The detailed performance comparisons are presented in Table~\ref{tab: exp_smoothing}, further underscoring the effectiveness of our approach in enhancing network resilience. }

\textcolor{black}{Overall, our experiments with ISTA-Net and E2E-VarNet provide additional evidence that SMUG and Weighted SMUG can be successfully integrated into different unrolled network architectures. These findings highlight the generalizability of our method and suggest its potential for improving robustness in image reconstruction tasks.}

\begin{figure*}
    \centering
    \includegraphics[width=0.88\textwidth]{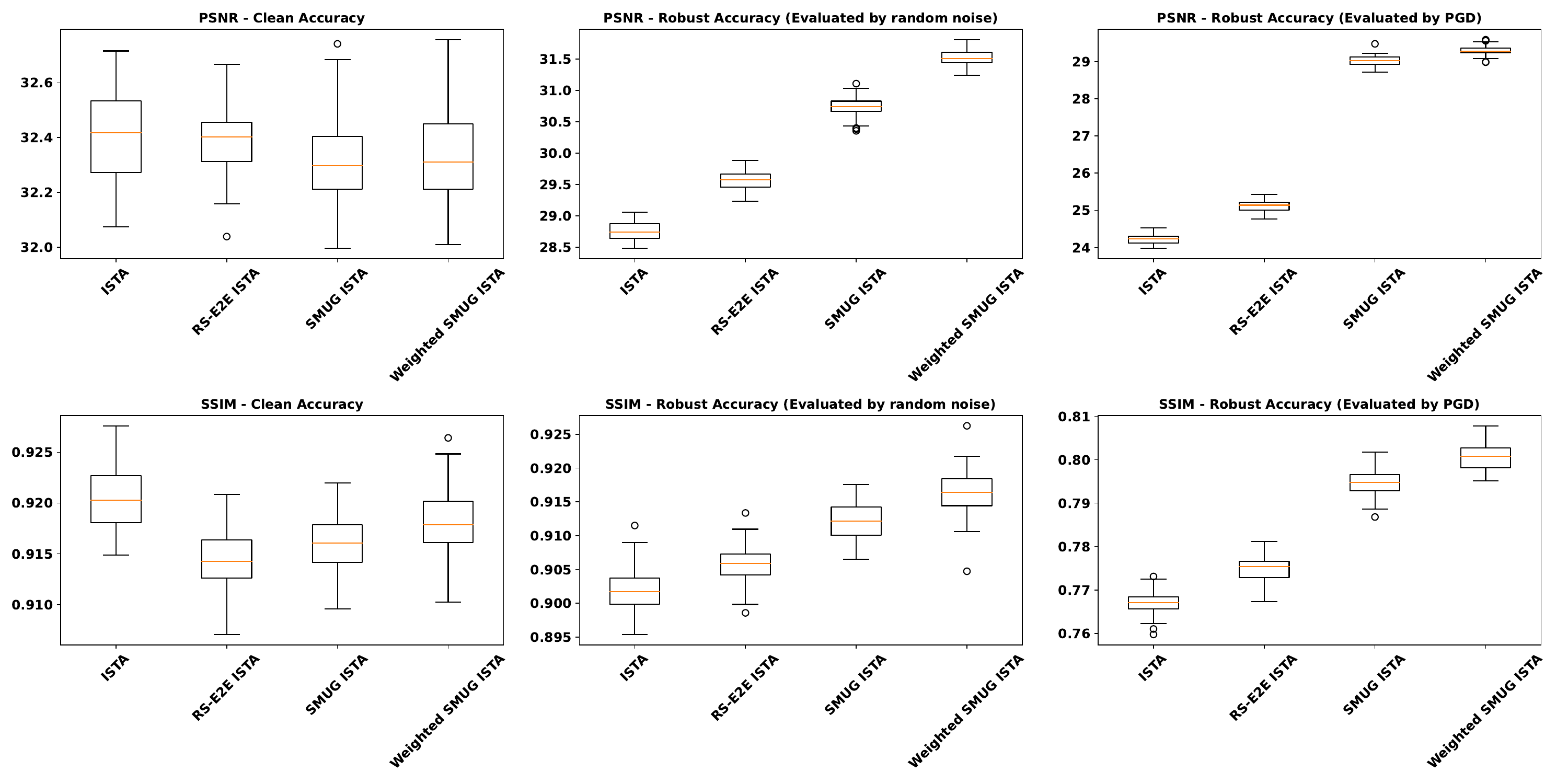}
    \vspace{-0.35cm}
    \caption{{ \textcolor{black}{Reconstruction accuracy box plots for the fastMRI \textbf{knee} dataset with 4x acceleration factor for the case of \textbf{ISTA-Net}.} \textcolor{black}{The additive random Gaussian  noise in the second column plots is obtained using a standard deviation of $0.01$. The worst-case additive noise in the third  column is obtained using the PGD  method with $\epsilon = 0.02$}.}}
    \label{fig:box_plot_knee_ISTA_net}
    \vspace{-0.01 in}
\end{figure*}


 \begin{table}[htb]
\centering
\caption{{
\textcolor{black}{Accuracy performance of different smoothing architectures \ref{eq: denoised smoothing mri}, {\us}, and Weighted {\us} together with the vanilla E2E-VarNet. Here `Clean Accuracy', `Noise Accuracy', and `Robust Accuracy' refer to PSNR/SSIM evaluated on benign data, random noise-injected data, and PGD attack-enabled adversarial data, respectively.
$\uparrow$ indicates that a higher number is a better reconstruction accuracy. The result $a${\tiny{$\pm b$}} represents mean $a$ and standard deviation $b$ over {64} testing images.
}}}
\label{tab: exp_smoothing}
\vspace*{-2mm}
\resizebox{0.48\textwidth}{!}{%
\begin{tabular}{c|cc|cc|cc}
\toprule[1pt]
\midrule
Models 
& \multicolumn{2}{c}{Clean Accuracy} 
& \multicolumn{2}{c}{Noise Accuracy} 
& \multicolumn{2}{c}{Robust Accuracy} 
\\
Metrics 
& PSNR \textcolor{red}{$\uparrow$} & SSIM \textcolor{red}{$\uparrow$}
& PSNR \textcolor{red}{$\uparrow$} & SSIM \textcolor{red}{$\uparrow$}
& PSNR \textcolor{red}{$\uparrow$} & SSIM \textcolor{red}{$\uparrow$}
\\
\midrule
Vanilla E2E-VarNet
& 32.83\footnotesize{$\pm$0.24}
& 0.912\footnotesize{$\pm$0.05}
& 30.15\footnotesize{$\pm$0.37}
& 0.882\footnotesize{$\pm$0.07}
& 23.78\footnotesize{$\pm$0.52}
& 0.742\footnotesize{$\pm$0.07}
\\
\midrule
RS-E2E
& 32.58\footnotesize{$\pm$0.37}
& 0.904\footnotesize{$\pm$0.03}
& 30.67\footnotesize{$\pm$0.42}
& 0.889\footnotesize{$\pm$0.04}
& 24.56\footnotesize{$\pm$0.32}
& 0.771\footnotesize{$\pm$0.08}
\\
{\us}  
& 32.64\footnotesize{$\pm$0.27}
& 0.907\footnotesize{$\pm$0.08}
& 31.24\footnotesize{$\pm$0.39}
& 0.895\footnotesize{$\pm$0.04}
& 27.85\footnotesize{$\pm$0.38}
& 0.821\footnotesize{$\pm$0.056}
\\
\rowcolor[gray]{.8}
Weighted {\us}
& 32.78\footnotesize{$\pm$0.34}
& 0.909\footnotesize{$\pm$0.067}
& 31.43\footnotesize{$\pm$0.44}
& 0.899\footnotesize{$\pm$0.05}
& 28.26\footnotesize{$\pm$0.41}
& 0.831\footnotesize{$\pm$0.067}
\\
\midrule
\bottomrule[1pt]
\end{tabular}%
}
\vspace*{-4mm}
\end{table}

\section{Discussion and Conclusion}
\label{sec:conclusion}
In this work, we proposed a scheme for improving the robustness of DL-based MRI reconstruction. In particular, we investigated deep unrolled reconstruction's 
weaknesses in robustness against worst-case or noise-like additive perturbations, sampling rates, and unrolling steps. To improve the robustness of the unrolled scheme, we proposed {\us} with a novel unrolled smoothing loss. 
We also provided a theoretical analysis on the robustness achieved by our proposed method integrated into MoDL.
 Compared to the vanilla
{\modl} approach and other schemes,
we empirically showed that our approach is effective and can significantly improve the robustness of 
a deep unrolled scheme
against a diverse set of external perturbations.
We also further improved {\us}'s robustness by introducing weighted smoothing as an alternative to conventional RS, which adaptively weights different images when aggregating them. \textcolor{black}{While we applied the proposed smoothing schemes to several unrolled deep image reconstruction models such as MoDL, ISTA-Net, and VarNet, we hope to study applicability to other deep network models in future work.}
\textcolor{black}{We also plan to apply the proposed schemes to other imaging modalities and evaluate robustness against additional types of realistic perturbations. While we theoretically characterized the robustness error for SMUG, we hope to further analyze its accuracy-robustness trade-off with perturbations.}


\appendices

\section{Proof of Theorem 1}
\label{sec: Appendix}

\subsection{Preliminary of Theorem 1}
\begin{lemma}
Let  $f : \mathbb{R}^d \to \mathbb{R}^m $ be any bounded function. Let \( \boldsymbol\eta \sim \mathcal{N}(0, \sigma^2 \mathbf{I}) \). We define \( g : \mathbb{R}^d \to \mathbb{R}^m \) as
\[ g(\x) = \mathbb{E}_{\boldsymbol\eta} [f(\x + \boldsymbol\eta)]. \]
Then, \( g \) is an \( \frac{M}{\sqrt{2\pi}\sigma} \)-Lipschitz map, where \( M = 2\max_{\x\in\R^d} (\|f(\x)\|_2)\). In particular, for any \( \x, \boldsymbol\delta \in \mathbb{R}^d \):
\[ \|g(\x) - g(\x + \boldsymbol \delta)\|_2 \leq  \frac{M}{\sqrt{2 \pi}\sigma} \|\boldsymbol\delta\|_2. \]
\end{lemma}
\begin{proof}
The proof of this bound follows recent work~\cite{wolfmaking}, with a modification on $M$.
Let $\mu$ be the probability distribution function of random variable $\boldsymbol\eta$. By the change of variables $\boldsymbol w=\x+\boldsymbol\eta$
and $\boldsymbol w=\x+\boldsymbol\eta + \boldsymbol\delta$ for the integrals constituting $g(\x)$ and $g(\x+\boldsymbol \delta)$, 
we have $\|g(\x)-g(\x+\boldsymbol \delta)\|_2=\|\int_{\mathbb{R}^d} f(\boldsymbol w)[\mu(\boldsymbol w-\x)-\mu(\boldsymbol w-\x-\boldsymbol \delta)]\ d\boldsymbol w\|_2$. Then, we have $\|g(\x)-g(\x+\boldsymbol\delta)\|_2$
\begin{align*}
    \leq\int_{\mathbb{R}^d} \|f(\boldsymbol w)[\mu(\boldsymbol w-\x)-\mu(\boldsymbol w-\x-\boldsymbol \delta)]\|_2\ d\boldsymbol w,
\end{align*}
which is a standard result for the norm of an integral.
We further apply Holder's inequality to upper bound $\|g(\x)-g(\x+\boldsymbol\delta)\|_2$ with
\begin{align}
    \max_{\x\in\R^d}(\|f(\x)\|_2) \int_{\R^d}|\mu(\boldsymbol w-\x)-\mu(\boldsymbol w-\x-\boldsymbol \delta)\ |d\boldsymbol w.
\end{align}
Observe that $\mu(\boldsymbol w -\x)\geq \mu(\boldsymbol w -\x-\boldsymbol\delta)$ if $\|\boldsymbol w -\x\|_2\leq \|\boldsymbol w-\x-\boldsymbol\delta\|_2$. Let $D = \{\boldsymbol w: \|\boldsymbol w -\x\|_2\leq \|\boldsymbol w-\x-\boldsymbol\delta\|_2\}$. Then, we can rewrite the above bound as
\begin{align}
    &=\max_{\x\in\R^d}(\|f(\x)\|_2)\cdot2\int_D [\mu(\boldsymbol w-\x)-\mu(\boldsymbol w-\x-\boldsymbol \delta)]\ d\boldsymbol{w}\\
    &= \frac{M}{2}
    \begin{pmatrix}
2\int_D\mu(\boldsymbol w-\x)\ d\boldsymbol{w}-2\int_D\mu(\boldsymbol w-\x-\boldsymbol \delta)\ d\boldsymbol{w}.
\end{pmatrix}
\end{align}
Following Lemma~3 in \cite{lakshmanan2008decentralized}, we obtain the bound
\begin{align}
     2\int_D\mu(\boldsymbol w-\x)\ d\boldsymbol{w}-2\int_D\mu(\boldsymbol w-\x-\boldsymbol \delta)\ d\boldsymbol{w}\leq \frac{2}{\sqrt{2\pi}\sigma}\|\boldsymbol\delta\|_2\:,
\end{align}
which implies that 
    $\|g(\x)-g(\x+\boldsymbol \delta)\|_2\leq \frac{2\max_{\x\in\R^d}(\|f(\x)\|_2)}{\sqrt{2\pi}\sigma}\|\boldsymbol\delta\|_2 = \frac{M}{\sqrt{2\pi}\sigma}\|\boldsymbol\delta\|_2$. This completes the proof. 
\end{proof}
\subsection{Proof of Theorem 1}

\begin{proof}
Assume that the data consistency step in MoDL at iteration $n$ is denoted by $\mathbf{x}^{n}_{\text{M}}(\mathbf{A}^H \mathbf{y})$. 
We will sometimes drop the input and $\mathbf{y}$ dependence for notational simplicity.
Then
    \begin{align}
    &\mathbf{x}^1_{\text{M}} = (\mathbf{A}^H\mathbf{A}+\mathbf{I})^{-1}(\mathbf{A}^H\mathbf{y}+\mathcal{D}_{\boldsymbol\theta}(\mathbf{A}^H\mathbf{y}))\:,\\
    &\mathbf{x}^n_{\text{M}} = (\mathbf{A}^H\mathbf{A}+\mathbf{I})^{-1}(\mathbf{A}^H\mathbf{y}+\mathcal{D}_{\boldsymbol\theta}(\mathbf{x}^{n-1}_{\text{M}})), 
    \end{align} 
    where $\mathcal{D}_{\boldsymbol\theta}$ is the denoiser function. For the sake of simplicity and consistency with the experiments, we use the weighting parameter $\lambda=1$ (in the data consistency step). We note that the proof works for arbitrary $\lambda$. 
    SMUG introduces an iteration-wise smoothing step into MoDL as follows:
    \begin{align}
    &\mathbf{x}_{\text{S}}^1 = ((\mathbf{A}^H\mathbf{A}+\mathbf{I})^{-1} (\mathbf{A}^H\mathbf{y} + \mathbb{E}_{\boldsymbol{\eta}_1}[  \mathcal{D}_{\boldsymbol\theta}(\mathbf{A}^H\mathbf{y} + \boldsymbol{\eta}_1) ])\\
    &\mathbf{x}_{\text{S}}^{n} = ((\mathbf{A}^H\mathbf{A}+\mathbf{I})^{-1} (\mathbf{A}^H\mathbf{y} + \mathbb{E}_{\boldsymbol\eta_{n}}[  \mathcal{D}_{\boldsymbol\theta}(\mathbf{x}_{\text{S}}^{n-1} + \boldsymbol\eta_{n})] ) \label{eq22}\\
    &=  (\mathbf{A}^H\mathbf{A} + \mathbf{I})^{-1} (\mathbf{A}^H \mathbf{y}) +\\ &(\mathbf{A}^H\mathbf{A} + \mathbf{I})^{-1}\mathbb{E}_{\boldsymbol\eta_n}[\mathcal{D}_{\boldsymbol\theta}( \mathbf{x}_{\text{S}}^{n-1}+\boldsymbol{\eta}_n)]  \notag,
    \end{align}
where we apply the expectation to the denoiser $\mathcal{D}_{\boldsymbol\theta}$ at each iteration. We use $\boldsymbol{\eta}_n$ to denote the noise during smoothing at iteration $n$.
The robustness error of SMUG after $n$ iterations is $\|\mathbf{x}_{\text{S}}^n(\mathbf{A}^H\mathbf{y})-\mathbf{x}_{\text{S}}^n(\mathbf{A}^H(\mathbf{y}+\boldsymbol{\delta}))\|$. We apply Lemma~1 and properties of the norm (e.g., triangle inequality) to bound $\|\mathbf{x}_{\text{S}}^n(\mathbf{A}^H\mathbf{y})-\mathbf{x}_{\text{S}}^n(\mathbf{A}^H(\mathbf{y}+\boldsymbol{\mathbf{\delta}}))\|$ as
\begin{align}
     &\leq \|(\mathbf{A}^H\mathbf{A}+\mathbf{I})^{-1}\mathbf{A}^H\boldsymbol\delta\|   \label{eq24} \\
     &+\|(\mathbf{A}^H\mathbf{A}+\mathbf{I})^{-1}\cdot\bigl(\E_{\boldsymbol\eta_n}[\mathcal{D}_{\boldsymbol\theta}\bigl(\mathbf{x}_{\text{S}}^{n-1}(\mathbf{A}^H\mathbf{y})+\boldsymbol\eta_n\bigr)]-\notag\\
     &\E_{\boldsymbol\eta_n}[\mathcal{D}_{\boldsymbol\theta}\bigl(\mathbf{x}_{\text{S}}^{n-1}(\mathbf{A}^H(\mathbf{y}+\boldsymbol\delta))+\boldsymbol\eta_n\bigr)]\bigr)\|\notag\\
       & \leq \|(\mathbf{A}^H\mathbf{A}+\mathbf{I})^{-1} \|_2 \|\mathbf{A}^H\boldsymbol\delta\|_2 \notag \\
     & +  \|(\mathbf{A}^H\mathbf{A}+\mathbf{I})^{-1}\|_2 \|\E_{\boldsymbol\eta_n}[\mathcal{D}_{\boldsymbol\theta}\bigl(\mathbf{x}_{\text{S}}^{n-1}(\mathbf{A}^H\mathbf{y})+\boldsymbol\eta_n\bigr)]-\notag\\
     &\E_{\boldsymbol\eta_n}[\mathcal{D}_{\boldsymbol\theta}\bigl(\mathbf{x}_{\text{S}}^{n-1}(\mathbf{A}^H(\mathbf{y}+\boldsymbol\delta))+\boldsymbol\eta_n\bigr)]\|\notag\\
     &\leq 
 \|(\mathbf{A}^H\mathbf{A}+\mathbf{I})^{-1} \|_2 \|\mathbf{A}^H\boldsymbol\delta\|_2 +  \|(\mathbf{A}^H\mathbf{A}+\mathbf{I})^{-1}\|_2 \times \\ \notag
     &
     \begin{pmatrix}
\frac{M}{\sqrt{2\pi}\sigma}
\end{pmatrix}
     \|\mathbf{x}_{\text{S}}^{n-1}(\mathbf{A}^H\mathbf{y})-\mathbf{x}_{\text{S}}^{n-1}(\mathbf{A}^H(\mathbf{y}+\boldsymbol{\mathbf{\delta}}))\|. 
\end{align}
Here, $M = 2\max_{\x}(\|\mathcal{D}_{\boldsymbol\theta}(\x)\|)$. 
Then we plug in the expressions for $\mathbf{x}_{\text{S}}^{n-1}(\mathbf{A}^H\mathbf{y})$ and $\mathbf{x}_{\text{S}}^{n-1}(\mathbf{A}^H(\mathbf{y}+\boldsymbol{\mathbf{\delta}}))$ (from \eqref{eq22}) and bound their normed difference with
 $\|(\mathbf{A}^H\mathbf{A}+\mathbf{I})^{-1}\mathbf{A}^H\boldsymbol\delta\| + 
\|(\mathbf{A}^H\mathbf{A}
 +\mathbf{I})^{-1}\cdot\bigl(\E_{\boldsymbol\eta_{n-1}}[\mathcal{D}_{\boldsymbol\theta}\bigl(\mathbf{x}_{\text{S}}^{n-2}(\mathbf{A}^H\mathbf{y})+\boldsymbol\eta_{n-1}\bigr)]-\E_{\boldsymbol\eta_{n-1}}[\mathcal{D}_{\boldsymbol\theta}\bigl(\mathbf{x}_{\text{S}}^{n-2}(\mathbf{A}^H(\mathbf{y}+\boldsymbol\delta))+\boldsymbol\eta_{n-1}\bigr)]\bigr)\|$. This is bounded above similarly as for \eqref{eq24}.
 We repeat this process until we reach the initial $\mathbf{x}_{\text{S}}^{0}$ on the right hand side. This yields the following bound involving a geometric series.
\begin{align}
& \|\mathbf{x}_{\text{S}}^n(\mathbf{A}^H\mathbf{y})-\mathbf{x}_{\text{S}}^n(\mathbf{A}^H(\mathbf{y}+\boldsymbol{\mathbf{\delta}}))\| \\
    &\leq \|\mathbf{A}^H\boldsymbol\delta\|_2 
\begin{pmatrix}
\sum_{j=1}^n \|(\mathbf{A}^H\mathbf{A}+\mathbf{I})^{-1}\|_2^j\cdot\begin{pmatrix}
\frac{M}{\sqrt{2\pi}\sigma}
\end{pmatrix}^{j-1}
\end{pmatrix} \notag\\
     & + \|(\mathbf{A}^H\mathbf{A}+\mathbf{I})^{-1}\|_2^n
    \begin{pmatrix}
\frac{M}{\sqrt{2\pi}\sigma}
\end{pmatrix}^{n}  \|\mathbf{A}^H\boldsymbol\delta\|_2 \\
& \hspace{-0.1in} \leq \|\mathbf{A}\|_2\|\boldsymbol\delta\|_2
\|(\mathbf{A}^H\mathbf{A}+\mathbf{I})^{-1}\|_2  \begin{pmatrix}
\frac{1-\begin{pmatrix}
\frac{M}{\sqrt{2 \pi}\sigma}
\end{pmatrix}^{n}  \|(\mathbf{A}^H\mathbf{A}+\mathbf{I})^{-1}\|_2^{n}  }{1 - \frac{M}{\sqrt{2 \pi}\sigma} \|(\mathbf{A}^H\mathbf{A}+\mathbf{I})^{-1}\|_2}
\end{pmatrix} \notag \\
     &+ \|(\mathbf{A}^H\mathbf{A}+\mathbf{I})^{-1}\|_2^n
    \begin{pmatrix}
\frac{M}{\sqrt{2\pi}\sigma}
\end{pmatrix}^{n}\|\mathbf{A}\|_2\|\boldsymbol\delta\|_2 \leq C_n \|\boldsymbol\delta\|_2,
\end{align}
where we used the geometric series formula, and $C_n = \alpha \|\mathbf{A}\|_2     \begin{pmatrix}
\frac{1-\begin{pmatrix}
\frac{M \alpha}{\sqrt{2 \pi}\sigma}
\end{pmatrix}^{n}   }{1 - \frac{M \alpha}{\sqrt{2 \pi}\sigma} }
\end{pmatrix} +  \|\mathbf{A}\|_2   \begin{pmatrix}
\frac{M \alpha}{\sqrt{2 \pi}\sigma}
\end{pmatrix}^{n}$, with $\alpha =  \|(\mathbf{A}^H\mathbf{A}+\mathbf{I})^{-1}\|_2$.
\end{proof}
                                                                                                                                                                                                                                                                                                                                                                                                                                                 


\vspace{-0.1in}
\bibliographystyle{IEEEbib}
\bibliography{reference,refs_adv}

\begin{IEEEbiography}[{\includegraphics[width=1in,height=1.25in,clip,keepaspectratio]{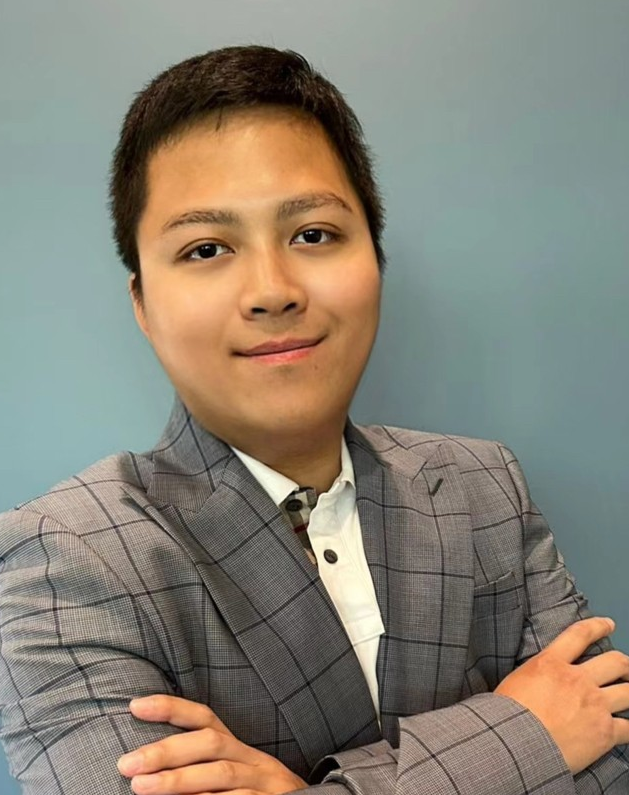}}]{Shijun Liang}(Member, IEEE)  received his B.S. degree in Biochemistry from the University of California, Davis, CA, USA, in 2017. In 2025, he received Ph.D. student in the Department of Biomedical Engineering at Michigan State University, East Lansing, MI, USA. His research focuses on machine learning and optimization techniques for solving inverse problems in imaging. Specifically, he is interested in machine learning based image reconstruction and in enhancing the robustness of learning-based reconstruction algorithms.
\end{IEEEbiography}

\begin{IEEEbiography}[{\includegraphics[width=1in,height=1.25in,clip,keepaspectratio]{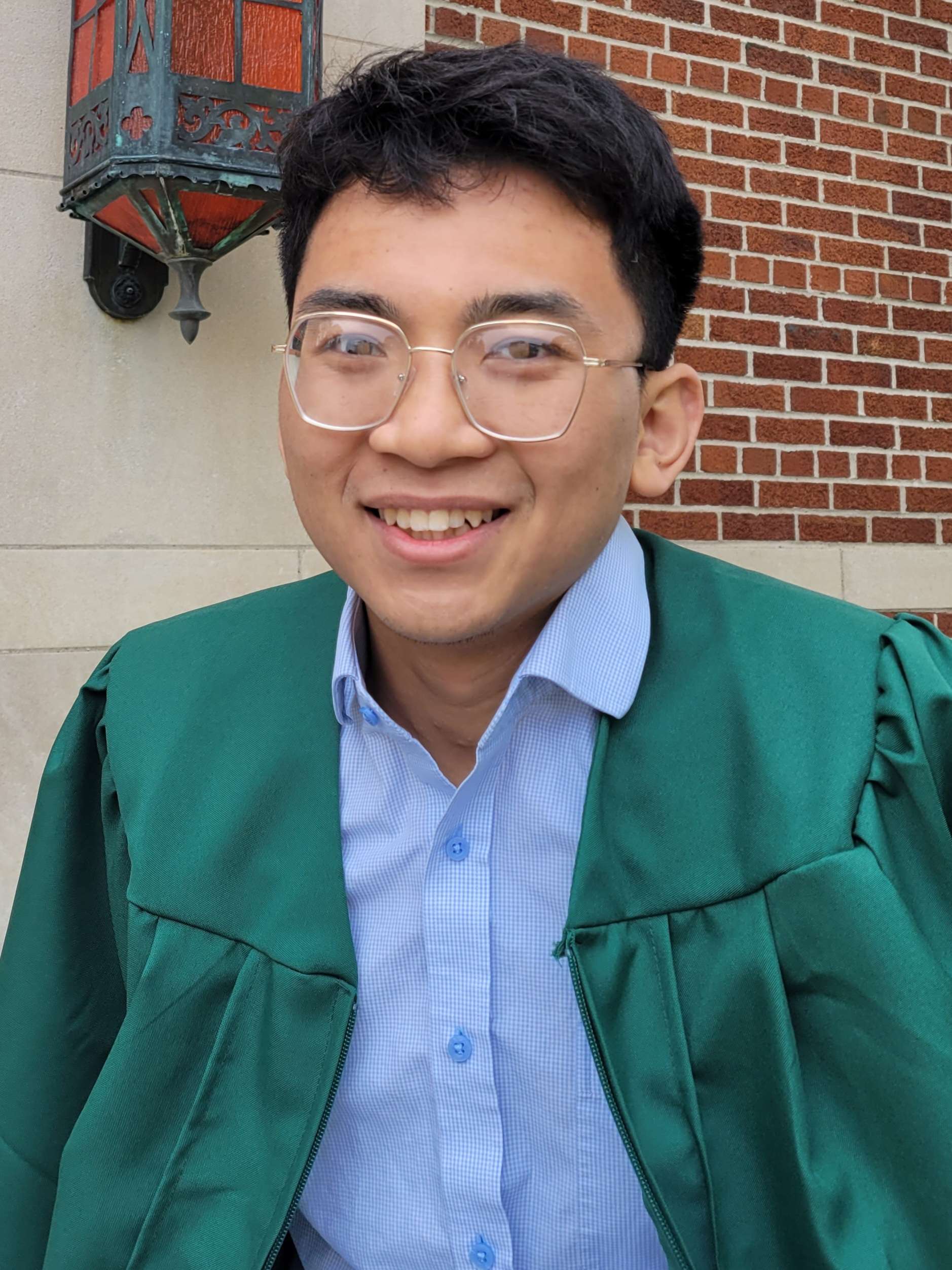}}]{Minh Van-Hoang Nguyen}
    recieved his B.A. in Mathematics from Michigan State University (MSU) in 2025. From fall 2025, he is a Ph.D. student in Applied and Computational Mathematics at the California Institute of Technology. While at MSU, he worked in machine learning, optimal transport and PDEs in the Computational Mathematics, Science (CMSE) and Engineering Department and the Mathematics Department. He is a recipient of the Outstanding Alumni Research Award at the 2025 MSU's CMSE 10th Anniversary workshop and several undergraduate awards from the Mathematics Department.
\end{IEEEbiography}

\begin{IEEEbiography}[{\includegraphics[width=1in,height=1.25in,clip,keepaspectratio]{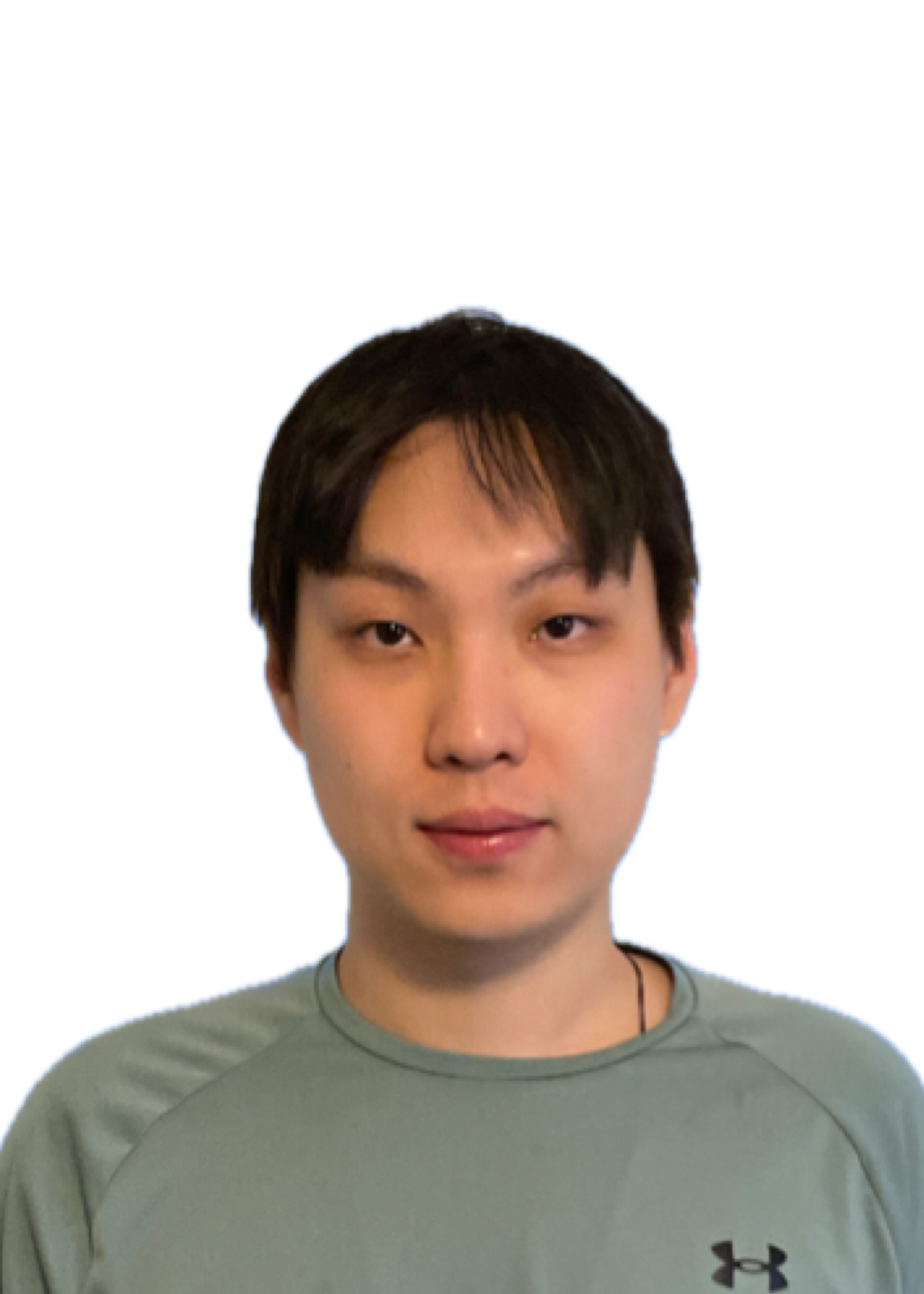}}]{Jinghan Jia}
(Student Member, IEEE) is a Ph.D. candidate in Computer Science at Michigan State University. His research focuses on trustworthy and efficient foundation models, including machine unlearning, reinforcement learning from human feedback (RLHF), and optimization for large language and diffusion models. He has co-authored over 30 papers, with 11 as first author, at top venues such as NeurIPS, ICLR, ICML, EMNLP, and CVPR.
\end{IEEEbiography}
\begin{IEEEbiography}[{\includegraphics[width=1in,height=1.25in,clip,keepaspectratio]{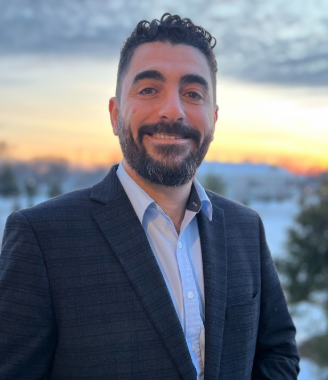}}]{Ismail R. Alkhouri}(Member, IEEE) is a Research Scientist at SPA Inc., providing technical support to the Information Innovative Office at the Defense Advanced Research Projects Agency (DARPA). He is also a Research Scholar at the University of Michigan (Electrical Engineering and Computer Science (EECS) Department) and Michigan State University (Computational Mathematics, Science, and Engineering (CMSE) Department). He received a Ph.D. in Electrical and Computer Engineering from the University of Central Florida in May 2023. From 2019 to 2022, he was a research intern at the Air Force Research Laboratory (Information directorate), and from July 2023 to December 2024, he was a Postdoctoral Researcher at Michigan State University (CMSE Department) and the University of Michigan (EECS Department). He is a recipient of the Rising Stars Award at the 2025 Conference on Parsimony and Learning (CPAL). He is also a recipient of the Outstanding Alumni Research Award at the 2025 MSU's CMSE 10th Anniversary workshop. His research focuses on computational imaging with deep generative models and differentiable methods for combinatorial optimization. His work was recognized as a finalist for best paper awards at ICASSP 2021 and MLSP 2023.
\end{IEEEbiography}

\begin{IEEEbiography}[{\includegraphics[width=1in,height=1.25in,clip,keepaspectratio]{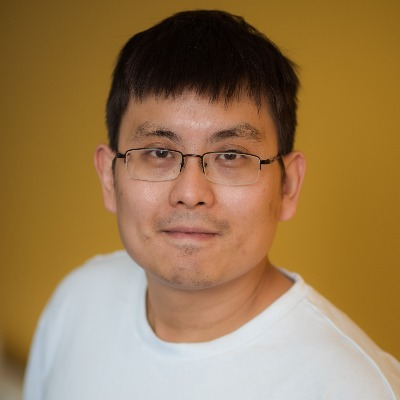}}]{Sijia Liu } (Senior Member, IEEE) 
is an Associate Professor in Computer Science and Engineering at Michigan State University, and an Affiliated Professor at IBM Research. His research focuses on scalable and trustworthy AI, spanning machine unlearning for vision and language models, scalable optimization, adversarial robustness, and data-model efficiency. He has received several honors, including the NSF CAREER Award (2024), INNS Aharon Katzir Young Investigator Award (2024), MSU Withrow Rising Scholar Award (2025), Best Paper Runner-Up Award at UAI (2022), and Best Student Paper Award at ICASSP (2017). He is a Senior Member of IEEE and serves on the IEEE Signal Processing Society’s Machine Learning for Signal Processing Technical Committee. He is also an Associate Editor for both the IEEE Transactions on Signal Processing and the IEEE Transactions on Aerospace and Electronic Systems. He is the co-founder of the New Frontiers in Adversarial Machine Learning Workshop series (ICML/NeurIPS 2021--2024) and has delivered numerous tutorials on trustworthy and scalable ML at major conferences such as ICASSP, AAAI, CVPR, and NeurIPS.
\end{IEEEbiography}

\begin{IEEEbiography}
[{\includegraphics[width=0.95\textwidth]{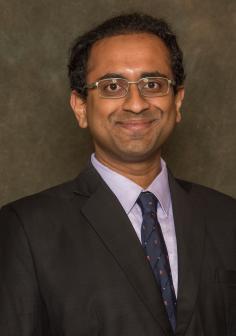}}]{Saiprasad Ravishankar} (Senior Member, IEEE) received the B.Tech. degree in Electrical Engineering
from the Indian Institute of Technology Madras,
Chennai, India, in 2008, and the M.S. and Ph.D. degrees in Electrical and Computer Engineering from the
University of Illinois at Urbana-Champaign, Urbana, IL, USA, in 2010 and 2014, respectively. He was
an Adjunct Lecturer and a Postdoctoral Research Associate with the University
of Illinois at Urbana-Champaign from February to August, 2015. Since August 2015, he was a Postdoc with the
Department of Electrical Engineering and Computer Science at the University
of Michigan, Ann Arbor, MI, USA, and then a Postdoc Research Associate with the Theoretical
Division at Los Alamos National Laboratory, Los Alamos, NM, USA, from
August 2018 to February 2019. 
He is currently an Associate Professor with
the Departments of Computational Mathematics, Science and Engineering, and Biomedical Engineering,
Michigan State University (MSU), Michigan, USA. 
His research interests include biomedical and computational imaging, machine learning, signal and image
processing, inverse
problems, neuroscience, and large-scale data processing and optimization. He was the recipient
of the IEEE Signal Processing Society Young Author Best Paper Award in 2016. A paper he co-authored won a
Best Student Paper Award at the IEEE International Symposium on Biomedical
Imaging (ISBI) 2018 and other papers were award
finalists at the IEEE International Workshop on Machine Learning for Signal
Processing (MLSP) 2017, ISBI 2020, and Optica Imaging Congress, 2023. He is currently a member of the IEEE Machine Learning for Signal Processing (MLSP) and Bio Imaging and Signal Processing (BISP) Technical Committees. He has organized several special
sessions and workshops on computational imaging and machine learning themes
including at the Institute for Mathematics and its Applications (IMA) in 2019, the Institute for Mathematical and Statistical Innovation (IMSI) in 2024, IEEE
Image, Video, and Multidimensional Signal Processing (IVMSP) Workshop
2016, MLSP 2017, ISBI 2018, the International Conference on Computer
Vision (ICCV) 2019 and 2021, etc.
\end{IEEEbiography}
\end{document}